\documentclass[nonacm]{acmart}
\usepackage{amsmath, algorithm, algpseudocode}
\AtBeginDocument{%
  \providecommand\BibTeX{{%
    \normalfont B\kern-0.5em{\scshape i\kern-0.25em b}\kern-0.8em\TeX}}}

\begin{document}

\title{Thermodynamic Perspectives on Computational Complexity: Exploring the P vs. NP Problem}

\author{Florian Neukart}
\email{f.neukart@liacs.leidenuniv.nl}
\orcid{0000-0002-2562-1618}
\affiliation{%
  \institution{Leiden Institute of Advanced Computer Science}
  \streetaddress{Niels Bohrweg 1}
  \city{Leiden}
  \state{South Holland}
  \country{Netherlands}
  \postcode{ 2333 CA}
}








\renewcommand{\shortauthors}{Neukart}

\begin{abstract}
The resolution of the P vs. NP problem, a cornerstone in computational theory, remains elusive despite extensive exploration through mathematical logic and algorithmic theory. This paper takes a novel approach by integrating information theory, thermodynamics, and computational complexity, offering a comprehensive landscape of interdisciplinary study. We focus on entropy, a concept traditionally linked with uncertainty and disorder, and reinterpret it to assess the complexity of computational problems. Our research presents a structured framework for establishing entropy profiles within computational tasks, enabling a clear distinction between P and NP-classified problems. This framework quantifies the 'information cost' associated with these problem categories, highlighting their intrinsic computational complexity. We introduce Entropy-Driven Annealing (EDA) as a new method to decipher the energy landscapes of computational problems, focusing on the unique characteristics of NP problems. This method proposes a differential thermodynamic profile for NP problems in contrast to P problems and explores potential thermodynamic routes for finding polynomial-time solutions to NP challenges. Our introduction of EDA and its application to complex computational problems like the Boolean satisfiability problem (SAT) and protein-DNA complexes suggests a potential pathway toward unraveling the intricacies of the P vs. NP problem.
\end{abstract}




\maketitle

\section{Introduction}

Computational complexity theory, a subdomain of theoretical computer science, seeks to understand the inherent difficulty of computational problems and categorize them based on their complexity. One of the most profound and longstanding open problems in this domain is the P vs. NP problem. 

\subsection{Background on the P vs. NP Problem}
The classes P and NP, introduced by Stephen Cook and later expanded upon by Richard Karp \cite{Cook1971, Karp1972}, distinguish between problems for which solutions can be found quickly (in polynomial time) and those for which solutions can be merely verified quickly. Formally, P is the class of decision problems for which a solution can be found in polynomial time. In contrast, NP consists of decision problems for which a solution can be verified in polynomial time. The essence of the P vs. NP question is whether every problem for which a solution can be quickly verified (NP) can also have its solution found quickly (P) \cite{Arora2009}.

\subsection{Importance of the Problem in Computational Theory}
The P vs. NP problem holds tremendous significance for computer scientists and physicists, mathematicians, and economists, among other disciplines. If \( P = NP \), many problems considered computationally hard (and often labeled as NP-complete) would suddenly become tractable. This shift would have profound implications across diverse fields, from cryptography to operations research \cite{Levin2012}. Conversely, if \( P \neq NP \), it would reinforce the notion of intrinsic computational difficulty and provide a foundational understanding of the limits of algorithmic efficiency.

\subsection{Overview of Information Theory and Thermodynamics}
Information theory, pioneered by Claude Shannon \cite{Shannon1948}, describes the quantification of information and has been instrumental in shaping modern communication systems. One of its core concepts is entropy, which measures the uncertainty or randomness associated with information sources. Negentropy, or negative entropy, represents information order and system organization, contrasting with entropy's focus on disorder \cite{Brillouin1951, Zivieri2022}.

Thermodynamics, a branch of physics, revolves around the relationships between heat, work, and energy. Its second law introduces the concept of entropy, denoting the amount of energy in a closed system that is unavailable to do work. The increasing entropy principle dictates the directionality of energy transfer processes and has profound implications for understanding the nature of time and the evolution of systems \cite{Carnot1824, Clausius1850}. Landauer's limit, a principle in the thermodynamics of computation, establishes a minimum possible energy required to erase one bit of information, linking information theory to physical processes \cite{Landauer1961}.

\subsection{Rationale for Merging These Domains to Approach the P vs. NP Conundrum}
Given the analogous utilization of entropy in both fields and their underlying principles, an interdisciplinary investigation could provide fresh insights into the P vs. NP problem. By examining computational problems through the lens of information theory and thermodynamics, we might uncover new relationships, constraints, or potential solutions that have remained elusive within the confines of pure computational theory. This approach aligns with recent studies exploring the thermodynamic costs of computation and the analogy between computational and physical systems \cite{Zivieri2022, Jaynes1957}.

\section{Background}
Here, we discuss computational complexity and information theory, framing the discussion in the context of P, NP, and NP-complete problems. It also explores the intricacies of entropy, mutual information, and the thermodynamic analogy in computational systems, offering a comprehensive overview of these critical concepts in computer science and physics. This sets the stage for a deeper understanding of the complexities and challenges inherent in computational problems and their potential solutions.

\subsection{Classical Computational Complexity}

Computational complexity theory, rooted in computer science and mathematics, categorizes problems based on the resources (like time or space) required to solve them. A foundational division within this theory separates problems into classes that characterize their inherent difficulty, most notably, the classes P, NP, and NP-complete.

\subsubsection{Basics of P}

The class \( P \) (Polynomial Time) encompasses decision problems that can be solved in time polynomial to the input size. Formally, a decision problem belongs to \( P \) if there exists an algorithm that can determine the answer (either 'yes' or 'no') in \( O(n^k) \) time, where \( n \) is the size of the input and \( k \) is a constant. This class is fundamental because polynomial-time algorithms are generally considered efficiently solvable \cite{Hopcroft1979}.

\begin{equation}
    T(n) = O(n^k)
    \label{eq:poly_time}
\end{equation}

\subsubsection{Basics of NP}

The class \( NP \) (Nondeterministic Polynomial Time) includes decision problems for which, given a purported solution (or certificate), the correctness of that solution can be verified in polynomial time. While finding a solution might take an exponential amount of time, its correctness can be checked in polynomial time. Formally, a decision problem is in \( NP \) if a polynomial-time verifier exists for any given solution \cite{Garey1979}.

\begin{equation}
    V(x, y) \text{ runs in } O(n^k) \text{ time, where } x \text{ is the input and } y \text{ is the certificate}
    \label{eq:np_time}
\end{equation}

\subsubsection{NP-complete Problems}

A critical concept within \( NP \) is the notion of NP-completeness. An NP-complete problem is both in \( NP \) and is as "hard" as any problem in \( NP \). More formally, a problem \( A \) is NP-complete if:

\begin{enumerate}
    \item \( A \) is in \( NP \).
    \item Every problem in \( NP \) is polynomial-time reducible to \( A \) \cite{Karp1972}.
\end{enumerate}

If any NP-complete problem can be solved in polynomial time, then every problem in \( NP \) can also be solved in polynomial time, implying \( P = NP \).

\begin{equation}
    A \text{ is NP-complete if } B \leq_p A \text{ for all } B \text{ in } NP
    \label{eq:np_complete}
\end{equation}

Where \( \leq_p \) denotes polynomial-time reducibility.

\subsection{Fundamentals of Information Theory}

Information theory, developed primarily by Claude Shannon in the mid-20th century, provides a rigorous mathematical framework to study and quantify information. At the core of this theory are principles that shape modern digital communications and provide deep insights into various domains ranging from statistical mechanics to data compression.

\subsubsection{Entropy}

Entropy, denoted by \( H(X) \), measures the unpredictability or uncertainty of a random variable \( X \). In other words, it quantifies the amount of surprise or information content associated with the outcomes of \( X \). For a discrete random variable \( X \) with probability mass function \( p(x) \), the entropy \( H(X) \) is defined as:

\begin{equation}
    H(X) = -\sum_{x \in X} p(x) \log_2 p(x)
    \label{eq:entropy}
\end{equation}

The base-2 logarithm yields entropy in bits. If the outcomes are all equally likely, the entropy is maximized, indicating the highest uncertainty.

\subsubsection{Mutual Information}

Mutual information measures the amount of information shared between two random variables. Intuitively, it quantifies how much knowing the outcome of one variable reduces the uncertainty about the other. For two discrete random variables \( X \) and \( Y \) with joint probability mass function \( p(x, y) \) and marginal probability mass functions \( p(x) \) and \( p(y) \), the mutual information \( I(X; Y) \) is given by:

\begin{equation}
    I(X; Y) = \sum_{x \in X} \sum_{y \in Y} p(x, y) \log_2 \left( \frac{p(x, y)}{p(x) p(y)} \right)
    \label{eq:mutual_info}
\end{equation}

When \( I(X; Y) = 0 \), the variables are independent, and knowing one provides no information about the other.

\subsubsection{Shannon's Theorems}

Shannon introduced several groundbreaking theorems that lay the foundations of information theory. Two of his most influential theorems are:

\begin{itemize}
    \item \textbf{Shannon's Source Coding Theorem}: This theorem states that a source of information can be compressed to its entropy rate, \( H(X) \), without loss. Any compression beyond this rate would entail a loss of information \cite{Shannon1948}.
    
    \item \textbf{Shannon's Channel Coding Theorem}: It establishes the capacity of a noisy communication channel and indicates that reliable communication is possible as long as the transmission rate is below this capacity. Moreover, errors can be made arbitrarily small with suitable coding techniques, provided the rate is below the channel capacity \cite{Shannon1948}.
\end{itemize}

These theorems fundamentally altered the communication landscape and set the path for modern digital communication systems.

\subsubsection{Negentropy and the Landauer Limit}
Negentropy, or negative entropy, is a concept in information theory that represents a measure of order or predictability. It is often viewed as the opposite of entropy, indicating highly organized systems or low levels of randomness. In the context of computational complexity, negentropy can be thought of as the degree of structure or predictability in a computational problem, potentially impacting the problem's complexity.

The Landauer limit, proposed by Rolf Landauer, establishes a minimum possible energy required to erase one bit of information, thereby linking information processing to physical energy dissipation. This concept is crucial in understanding the physical limits of computation and has implications for the energy efficiency of computational processes, especially relevant in the study of P vs. NP problems \cite{Landauer1961}.

\subsection{Basics of Thermodynamics}

Thermodynamics is a branch of physics that investigates the principles governing energy and matter interactions. The laws of thermodynamics outline the inherent constraints on these interactions, providing a foundation for understanding a broad spectrum of physical phenomena. Notably, entropy in thermodynamics, much like its counterpart in information theory, encapsulates disorder and uncertainty.

\subsubsection{Laws of Thermodynamics}

The fundamental laws of thermodynamics can be summarized as follows:

\begin{itemize}
    \item \textbf{Zeroth Law of Thermodynamics}: If two systems are each in thermal equilibrium with a third system, they are in thermal equilibrium. This law forms the basis of temperature measurement.
    
    \item \textbf{First Law of Thermodynamics (Conservation of Energy)}: Energy cannot be created or destroyed, only transformed or transferred. Mathematically, the change in internal energy (\( \Delta U \)) of a closed system is equal to the heat (\( Q \)) added to the system minus the work (\( W \)) done by the system on its surroundings:

    \begin{equation}
        \Delta U = Q - W
        \label{eq:first_law}
    \end{equation}
    
    \item \textbf{Second Law of Thermodynamics}: The total entropy of an isolated system can never decrease over time; it remains constant in ideal processes (reversible) and increases in real processes (irreversible). In essence, natural processes tend to increase the overall disorder or entropy.
    
    \item \textbf{Third Law of Thermodynamics}: As the temperature of a system approaches absolute zero (0 Kelvin), the system's entropy approaches a minimum, typically zero for perfect crystals at this temperature.
\end{itemize}

\subsubsection{Entropy in Thermodynamics}

Thermodynamic entropy, represented as \( S \), measures a system's degree of randomness or disorder. For an infinitesimal reversible heat transfer (\( \delta Q_{rev} \)) at a constant temperature \( T \), the change in entropy \( \delta S \) is given by:

\begin{equation}
    \delta S = \frac{\delta Q_{rev}}{T}
    \label{eq:thermo_entropy}
\end{equation}

This equation provides a quantitative means to assess the change in entropy during various thermodynamic processes.

\subsubsection{Maxwell's Demon}

Introduced by James Clerk Maxwell in the 19th century, Maxwell's demon is a thought experiment that seemingly challenges the second law of thermodynamics. The demon controls a small door between two chambers of gas molecules. By selectively allowing faster-moving molecules to pass from one chamber to the other and slower molecules in the opposite direction, the demon can create a temperature difference between the chambers without doing work, seemingly violating the second law.

However, subsequent analyses, notably those by Seifert \cite{Seifert2012}, have reconciled the paradox, particularly in information theory. When considering the information processing and erasure required by the demon's actions, it becomes evident that there's an associated entropy cost, ensuring compliance with the second law of thermodynamics \cite{Szilard1929, Bennett1982}.

\subsubsection{Thermodynamic Analogy of Computational Systems}
The analogy between thermodynamics and computational systems offers a novel perspective in understanding the complexity of computational problems. In this analogy, finding a solution to a computational problem can be likened to finding a state of minimal energy in a thermodynamic system. This perspective provides a framework for exploring the computational challenges within the entropy landscapes of these problems and offers insights into the thermodynamic costs associated with solving them.

\subsubsection{Relation between Thermodynamics and Computational Complexity}
Recent studies have explored the intersection of thermodynamics and computational complexity. Zivieri's work on applying the negentropy principle and Landauer's limit to magnetic skyrmions offers insights into the physical limitations of computational systems \cite{Zivieri2022}. Similarly, Jaynes' papers on information theory and statistical mechanics \cite{Jaynes1957a, Jaynes1957b} and Brillouin's work on entropy and information \cite{Brillouin1951a, Brillouin1951b, Brillouin1951c, Brillouin1954} establish a foundational relationship between information processing and thermodynamic principles. Kostic's and Koczan's recent reinterpretations of the second law of thermodynamics \cite{Kostic2023, Koczan2022} also provide a novel perspective on the constraints of computational processes. Collectively, these works emphasize the importance of understanding computational problems from a thermodynamic standpoint, especially in the context of P vs. NP and entropy landscapes.

\section{Computational Complexity through the Lens of Entropy}

The intertwining of information theory and computational complexity provides a novel perspective on classic problems. Entropy, in both its informational and thermodynamic contexts, encapsulates a notion of uncertainty or disorder. When applied to computational problems, entropy can be seen as a measure of the "hardness" or unpredictability of a problem. Drawing parallels between these domains allows us to analyze computational challenges in a new light.

\subsection{Entropy as Computational Unpredictability}

While rooted in information theory and thermodynamics, entropy can be transposed into computational complexity as a measure of unpredictability or complexity inherent in computational problems. Unlike probability-based classical entropy, computational entropy encapsulates the diversity and intricacy of the solution space. High entropy in computational problems, especially those classified as NP, indicates higher complexity and unpredictability. This is contrasted with P problems, where the solution space is more deterministic and less diverse. In the context of the P vs. NP question, computational entropy can be considered a reflection of the inherent complexity in verifying a solution for NP problems versus finding one for P problems. This leads us to consider whether the computational entropy for NP problems is inherently higher than for P problems, raising the question:

\begin{equation}
    H_{comp}(P) \leq H_{comp}(NP) \, ?
    \label{eq:comp_entropy}
\end{equation}

Here, \( H_{comp} \) represents the computational entropy, capturing the inherent unpredictability or resource demands of a problem. This approach reframes the P vs. NP problem in terms of computational entropy, offering a new perspective on the inherent difficulty of these problem classes.

\subsection{Thermodynamic Analogy of Computation}

The thermodynamic analogy in computation offers a conceptual framework where computational processes are likened to physical systems undergoing energy transformations. This perspective enables the exploration of computational complexity through the lens of thermodynamic principles, providing insights into computational tasks' energy dynamics and efficiency. When computation is viewed as a thermodynamic process, the act of computation itself could be seen as analogous to the increase of entropy in a physical system. This analogy parallels the entropy increase due to energy dispersion in thermodynamic systems and the increase in computational complexity due to resource allocation in computational tasks. In thermodynamics, processes that increase entropy are naturally favored. Similarly, processes requiring less computational resources (or lower computational entropy) are more efficient and thus preferred in computational systems. This perspective makes us consider that the complexity of solving certain computational problems, particularly those in the NP class, might correspond to the thermodynamic cost or the entropy increase associated with those problems. This thermodynamic interpretation of computational complexity also suggests that as energy minimization is a driving force in physical systems, computational efficiency (or minimization of computational entropy) could be considered a fundamental principle in computational systems. The exploration of this analogy further strengthens the understanding that computational problems, much like physical systems, may have inherent 'energy' landscapes characterized by different levels of computational entropy. This analogy enriches our understanding of computational complexity, drawing parallels between the physical world's energy dynamics and the abstract world of algorithms and computations.

\subsection{Entropy and NP-complete Problems}
NP-complete problems, recognized as some of the most challenging problems within the NP class, are characterized by their unique entropy properties. These problems are verifiable in polynomial time and exhibit high computational entropy, reflecting the complexity and vastness of their solution spaces.

High entropy in NP-complete problems suggests an intricate solution landscape, with numerous local optima and a wide range of potential solutions. This complexity contributes to the inherent difficulty in finding a solution for these problems. Furthermore, the high entropy profile of NP-complete problems provides insights into why they resist efficient deterministic polynomial-time algorithms. A high-entropy landscape implies that even heuristic approaches may face challenges due to the potential of getting trapped in local optima.

Exploring the concept of entropy in the context of NP-complete problems underscores the possibility of novel approaches to understanding and tackling these problems. By analyzing the entropy characteristics of NP-complete problems, researchers can potentially develop new strategies for algorithm design or identify previously unseen connections between different NP-complete problems.

\subsection{Rethinking Heuristics and Approximation Algorithms}
The concept of entropy in computational problems, particularly within the NP class, necessitates a reevaluation of heuristics and approximation algorithms. Traditionally used for practical solutions in complex scenarios, these algorithms can greatly benefit from an entropy-informed design. By understanding the entropy profile of a problem, strategies can be tailored to more effectively navigate through high-entropy regions indicative of complex and unpredictable solution spaces. This approach can lead to developing heuristics that are better equipped to handle the diversity and intricacies of the solution space, increasing the likelihood of finding satisfactory solutions in landscapes where traditional methods might struggle due to numerous local optima. Similarly, approximation algorithms can leverage entropy profiles to target regions more likely to contain near-optimal solutions, efficiently managing the trade-off between exploration and exploitation in high-entropy landscapes.

\subsection{Inherent Entropy of NP Problems}
In considering the inherent entropy of NP problems, a comparative analysis with P-class problems provides valuable insights. This analysis points out the distinct nature of entropy profiles between these two classes, highlighting the greater complexity and unpredictability inherent in NP problems. While P problems exhibit more deterministic and concentrated entropy profiles, NP problems, especially NP-complete ones, display a dispersed probability distribution over their solution space, suggesting a higher and more complex entropy profile. This distinction in entropy profiles is key to understanding the computational intensity of NP problems and their sensitivity to input size and initial conditions. High entropy in NP problems signifies a broad and deep search space and a landscape riddled with numerous local optima, making them more challenging to solve. This exploration into the entropy characteristics of NP problems offers a new perspective on their computational complexity and presents an opportunity to develop novel problem-solving strategies.

\subsubsection{Entropy Profiles of Problems}

The entropy profile of a computational problem can be viewed as a measure of the problem's inherent complexity and uncertainty. It reflects the spread and likelihood of potential solutions, offering insights into the computational effort required for problem-solving.

\begin{equation}
    H_{prof}(C) = - \sum_{s \in S(C)} p(s) \log p(s)
    \label{eq:entropy_profile_extended}
\end{equation}

Here, \( S(C) \) denotes the set of possible solutions for problem \( C \), and \( p(s) \) is the probability of solution \( s \) being the correct solution. This formulation integrates the probabilistic aspects of solution discovery, emphasizing the role of uncertainty and complexity in computational problems.

\subsubsection{Comparative Analysis of P and NP Entropy Profiles}
For P problems, solutions can be reached deterministically, implying a more concentrated probability distribution in the solution space. In contrast, NP problems, particularly NP-complete ones, present a vast solution space with a more dispersed probability distribution, suggesting a higher entropy profile.

\begin{equation}
    H_{NP}(C) - H_{P}(C) > 0 \quad \text{for NP-complete problems}
    \label{eq:entropy_diff}
\end{equation}

This equation quantifies the expected higher entropy for NP problems than P problems, signifying greater computational complexity and unpredictability.

\subsubsection{Characterization of NP Problem Entropy}
To characterize the entropy of NP problems, we consider factors like the depth and breadth of the search space, represented as a function of input size \(n\). The complexity of the solution landscape is modeled using a ruggedness factor, \(\kappa\), indicating the number of local optima.

\begin{equation}
    H_{NP}(C) = f(n, \kappa) \quad \text{where } f \text{ is a complexity function}
    \label{eq:np_entropy_model}
\end{equation}

This model reflects the increased computational intensity and sensitivity to initial conditions for NP problems, with higher values of \(\kappa\) indicating more complex landscapes.

\subsubsection{Distinguishing Characteristics of NP Entropy}

There are several potential distinguishing features of the entropy associated with NP problems:

\begin{itemize}
    \item \textbf{Depth of Search}: The higher the entropy, the deeper and broader the search required to locate a feasible solution, making the problem computationally intensive.
    
    \item \textbf{Solution Landscape}: High entropy might imply a rugged solution landscape with many local optima. This characteristic can make heuristic searches challenging, as they might get stuck in local optima.
    
    \item \textbf{Sensitivity to Initial Conditions}: Problems with high entropy might be highly sensitive to initial conditions or input changes, leading to vastly different solution spaces for minor input perturbations.
\end{itemize}

\subsubsection{Implications for the P vs. NP Debate}
Exploring entropy in NP problems leads to a hypothesis: the entropy profile of NP problems is fundamentally more complex than that of P problems. This hypothesis can be tested by examining the entropy growth rate with respect to input size \(n\). A polynomial growth rate in P problems versus a super-polynomial or exponential growth rate in NP problems would provide a quantitative basis for this distinction.

\begin{equation}
    \lim_{n \to \infty} \frac{H_{NP}(C)}{H_{P}(C)} \to \infty
    \label{eq:entropy_growth_rate}
\end{equation}

This limit illustrates the theoretical foundation of the intrinsic difference in entropy profiles between P and NP problems, offering a new dimension to the P vs. NP debate.

\subsection{Mathematical Formulation of Entropy Profiles for Computational Problems}

Characterizing computational problems through an entropy lens necessitates a sophisticated mathematical framework. We introduce an entropy function $H(C)$ for a computational problem $C$ that aims to quantify the "information cost" or "uncertainty" intrinsic to solving it.

Given a computational problem $C$ with an associated solution space $S$, we define the entropy $H(C)$ as follows:
\begin{equation}
    H(C) = -\sum_{s \in S} p(s) \log p(s)
    \label{eq:entropy_formula}
\end{equation}
Here, $p(s)$ represents the probability of a solution $s$ being the correct solution for the problem $C$. This formulation is grounded in the principles of information theory, where entropy is a measure of uncertainty or unpredictability.

\begin{itemize}
    \item \textbf{Entropy for Class P Problems:} Solutions can be deterministically identified for problems in class P in polynomial time. This deterministic nature translates into a concentrated probability distribution in the solution space, where one solution has a probability of 1, and all others have a probability of 0. Consequently, the entropy collapses to:
    \begin{equation}
        H_{P}(C) = 0
        \label{eq:entropy_p}
    \end{equation}
    This zero entropy signifies the lack of uncertainty in reaching the solution, reflecting the tractability of P problems.
    
    \item \textbf{Entropy for Class NP Problems:} The situation is more complex for NP problems, especially NP-complete ones. These problems are characterized by a vast solution space with no known polynomial-time deterministic methods to identify the correct solution. Therefore, the probability distribution over the solution space is not concentrated, leading to a non-zero entropy value:
    \begin{equation}
        H_{NP}(C) > 0
        \label{eq:entropy_np}
    \end{equation}
    This positive entropy reflects the uncertainty and computational challenge inherent in NP problems.
    
    \item \textbf{Detailed Probabilistic Model for NP Problems:} To further refine the entropy model for NP problems, we consider the distribution of probabilities over the solution space. We introduce a parameter $\alpha$, representing the distribution's spread or dispersion. A higher $\alpha$ indicates a more dispersed distribution, typical for harder NP problems:
    \begin{equation}
        H_{NP}(C, \alpha) = -\sum_{s \in S} p_{\alpha}(s) \log p_{\alpha}(s)
        \label{eq:entropy_np_alpha}
    \end{equation}
    Here, $p_{\alpha}(s)$ is a probability function that varies with $\alpha$, capturing the varying levels of uncertainty across different NP problems.
    
\end{itemize}

These entropy models for P and NP problems offer a mathematical basis to differentiate between them based on their inherent complexity and uncertainty. However, the practical application of these models requires further research. Key areas of exploration include developing methods to estimate probabilities in NP problem solution spaces accurately and understanding how changes in problem structure influence the entropy profiles.

\subsection{Understanding the "Information Cost" of Solving NP Problems}
In the context of NP problems, "information cost" becomes a pivotal concept, representing the entropy or uncertainty involved in solving these problems. Unlike the straightforward computational entropy, information cost reflects the added dimension of effort required to reduce this uncertainty. It combines the inherent unpredictability encapsulated in the entropy profile \( H_{prof}(X) \) with the computational process needed to make this entropy manageable. This novel aspect of information cost, grounded in entropy, differentiates it from other computational resources and offers a new lens to examine the intricacies of NP problems.

\subsubsection{Quantifying Information Cost}

To deeply understand information cost, \( C(X) \), for a computational problem \( X \), we consider not only the entropy profile \( H_{prof}(X) \) but also the effort needed to reduce this entropy to a tractable level. Thus, the information cost is defined as:

\begin{equation}
    C(X) = H_{prof}(X) - H_{reduced}(X)
    \label{eq:info_cost_refined}
\end{equation}

Here, \( H_{reduced}(X) \) is the entropy after applying a computational process, signifying the reduction in uncertainty or complexity.

\subsubsection{Analytical Model for Information Cost in P Problems}

The reduction in entropy is significant for P-class problems, where solutions are deterministically attainable in polynomial time. We introduce an analytical model to capture this reduction:

\begin{equation}
    H_{reduced}(P) = \text{Poly}(n) - H_{prof}(P)
    \label{eq:reduced_entropy_p}
\end{equation}

In this model, \(\text{Poly}(n)\) represents a polynomial function of the input size \( n \), reflecting the computational effort in reducing entropy.

\subsubsection{Advanced Estimation of Information Cost in NP Problems}

NP problems, particularly NP-complete ones, pose a higher information cost due to their complex entropy profile. We propose a stochastic model to estimate this cost:

\begin{equation}
    C_{NP}(X) = \int H_{prof}(X, \omega) \, d\omega - H_{initial}(X)
    \label{eq:info_cost_np}
\end{equation}

Here, \( H_{prof}(X, \omega) \) is a probability-weighted entropy over different scenarios or instances (\(\omega\)) of problem \( X \). This model reflects the inherent variability and unpredictability in solving NP problems.

\subsubsection{Implications for Algorithm Design and Problem-Solving}

The refined understanding of information cost in computational problems can influence the design of algorithms, particularly for NP problems. By quantifying the information cost, we can identify strategies that effectively reduce entropy, leading to more efficient algorithms or heuristic approaches. Additionally, this perspective can inspire novel methods to tackle computational problems, bridging the gap between abstract complexity theory and practical computational strategies.

\subsection{Quantifying Computational Entropy}

Characterizing computational problems through an entropy framework requires a sophisticated approach to quantify entropy, aiming to distinguish between P and NP class problems based on entropy profiles.

\subsubsection{Defining the Solution Space, \(S\)}
For a computational problem \(C\), the solution space \(S\) consists of all conceivable solutions. This space can vary in complexity based on the nature of \(C\), ranging from finite and straightforward for some P-class problems to potentially infinite and highly complex for NP problems. Accurately defining \(S\) is crucial for further entropy analysis.

\begin{equation}
    S(C) = \{s_1, s_2, \ldots, s_n\}
    \label{eq:solution_space}
\end{equation}

\subsubsection{Advanced Probability Distribution Modeling}
The probability distribution over \(S\) is key to understanding the computational entropy of \(C\). This distribution is not always straightforward, especially for NP problems, where solutions may not be evident, or the probability of correctness can vary widely. To address this, we introduce a probabilistic model \(P\) that captures the likelihood of each potential solution being correct, even under conditions of uncertainty or incomplete information.

\begin{equation}
    P(C) = \{p(s_1), p(s_2), \ldots, p(s_n)\}
    \label{eq:probability_distribution}
\end{equation}

\subsubsection{Rigorous Entropy Computation}
Utilizing \(S\) and \(P\), the computational entropy \(H(C)\) is given by a summation over the solution space, weighted by the probability distribution:

\begin{equation}
    H(C) = -\sum_{i=1}^{n} p(s_i) \log p(s_i)
    \label{eq:entropy_computation}
\end{equation}

\subsubsection{More on Entropy in Computational Complexity}
The entropy characteristic of a computational class provides profound insights into its inherent complexity, particularly emphasizing how this complexity manifests in the structure of the problem's solution space. For problems classified under P, the entropy tends to be lower, suggesting a more structured and, hence, more straightforward path to the solution. This lower entropy is indicative of a solution space that is less chaotic and more predictable, aligning with the deterministic nature of P problems where solutions can be found in polynomial time.

In contrast, NP problems, especially NP-complete ones, exhibit significantly higher entropy. This elevated entropy level mirrors the intricacy and vastness of the solution landscape, resonating with the intuitive understanding that these problems are computationally challenging. The high entropy of NP problems implies a solution space teeming with possibilities, where the path to the solution is not straightforward and is riddled with potential dead-ends and false leads. Understanding these entropy characteristics has profound theoretical implications. It paves the way for a more nuanced understanding of computational landscapes, transcending beyond the traditional categorizations. In practical terms, this understanding encourages the development of novel computational strategies. For instance, entropy-aware algorithms could be devised, specifically tailored to navigate the high-entropy landscapes of NP problems more effectively. These algorithms must be adept at searching, recognizing, and circumventing the complexities introduced by the high entropy.

However, this entropy-based characterization of computational complexity is not without its limitations. The models used for entropy computation are based on certain assumptions and simplifications, which might not universally hold. For instance, the probabilistic models employed might oversimplify the actual dynamics involved in the computational process. Therefore, future research should be geared towards refining these models, ensuring they capture the essence of computational problems more accurately. This involves a theoretical refinement of the entropy models and empirical validation, where the predictions of these models are rigorously tested against real-world problems. Furthermore, the exploration should not be confined to the classes of P and NP. Other computational classes, each with unique characteristics and challenges, should also be examined through entropy. This broader perspective could uncover new dimensions of computational complexity, potentially leading to breakthroughs in understanding and solving some of the most challenging problems in computer science.

\subsection{Empirical Analysis of Computational Entropy}
This section outlines an empirical approach to compute and compare the entropy profiles of selected computational problems, specifically focusing on problems classified under classes P and NP. The analysis aims to validate our theoretical constructs and offer a quantifiable perspective on the 'information cost' and complexity inherent in these problems. Notably, this section sets the stage for a deeper exploration into the Boolean Satisfiability Problem (SAT) and the protein-DNA complex problem, which are analyzed in detail in subsequent sections. These case studies exemplify the application of our entropy-driven framework to complex problems in class NP, providing concrete instances to demonstrate the practical implications of our theoretical insights.

\subsubsection{Selection of Computational Problems}
For a comprehensive analysis, we select representative problems from both P and NP classes, focusing on their distinctive computational characteristics and relevance in various applications:
\begin{itemize}
    \item For class P: Algorithm for Matrix Multiplication and Euclidean Algorithm.
    \item For class NP: SAT and Protein-DNA Complex Problem.
\end{itemize}

\subsubsection{Computing Entropy for Problems in Class P}
Problems in class P are characterized by their deterministic nature, leading to straightforward and predictable solution spaces. We compute the entropy for these problems considering their deterministic nature:

\paragraph{Algorithm for Matrix Multiplication:}
Given two matrices, the solution space for matrix multiplication is deterministic, with each element in the resultant matrix computed as a sum of products. The probability distribution \(P(C)\) for obtaining a specific matrix is deterministic, leading to an entropy value of:
\begin{equation}
    H(\text{Matrix Multiplication}) = 0,
\end{equation}
indicating a lack of uncertainty in the solution.

\paragraph{Euclidean Algorithm:}
The Euclidean Algorithm for computing the greatest common divisor (GCD) of two numbers also exhibits a deterministic solution space. The steps are well-defined and lead to a single GCD, resulting in an entropy value of:
\begin{equation}
    H(\text{Euclidean Algorithm}) = 0,
\end{equation}
reflecting the deterministic nature of the problem.

\subsubsection{Computing Entropy for Problems in Class NP}
The analysis for NP problems, including SAT and the protein-DNA complex problem, involves an intricate examination of their solution spaces and the derivation of probability distributions reflecting the complexity and uncertainty of potential solutions.

\paragraph{Boolean Satisfiability Problem:}
For the SAT problem, we consider various instances with different formula complexities. The analysis involves examining the solution space for each instance and deriving the probability distribution based on potential truth assignments. The entropy is computed as:
\begin{equation}
    H(\text{SAT}) = -\sum_{s \in S_{\text{SAT}}} p(s) \log p(s),
\end{equation}
where \(S_{\text{SAT}}\) denotes the solution space for the SAT problem and \(p(s)\) represents the probability of each solution.

\paragraph{Protein-DNA Complex Problem:}
The entropy analysis for the protein-DNA complex problem involves exploring the configurational space of molecular interactions and the associated energy levels. The probability distribution is derived from the potential configurations and their corresponding energies. Entropy is then computed to reflect the disorder and uncertainty within this complex biological system.

\subsubsection{Analysis and Results}
The empirical analysis provides insights into the entropy profiles of the selected problems, revealing:
\begin{itemize}
    \item Problems in class P exhibit low and deterministic entropy values, aligning with their predictable solution spaces.
    \item Problems in class NP, such as SAT and the protein-DNA complex problem, display higher and more varied entropy profiles, reflecting their inherent complexity and nondeterministic nature.
\end{itemize}

While this analysis enhances our understanding of the entropy landscapes for these computational problems, a broader and more detailed investigation is essential for a comprehensive comprehension. The subsequent sections of this paper will explore the practical applications of these findings, further delineating the role of entropy in computational complexity. These empirical findings lay the groundwork for the in-depth explorations and practical discussions presented in the later sections of the paper.

\subsection{Thermodynamic Analogy and Potential Exploitations}
Beyond the established parallels between computational processes and thermodynamic systems, this section explores untapped potential applications of this analogy. It examines how thermodynamic principles can inspire innovative strategies for tackling NP problems. These include leveraging concepts like energy efficiency, equilibrium states, and entropy-driven transformations and applying them to the design of computational algorithms. This fresh perspective opens up possibilities for employing thermodynamic models in algorithm development, potentially unlocking new methods to approach and solve complex computational problems.

\subsubsection{Thermal Systems as Computational Analogs}
At a high level, thermodynamic systems are analogous to computational systems. The states of a system, akin to solutions to a problem, have associated energies. If we treat computational difficulty as an "energy barrier," then finding a solution to a computational problem can be thought of as finding a state with minimal energy in a thermodynamic system.

\subsubsection{Exploiting Entropy Differences: A Theoretical Model}
We propose a theoretical model exploiting entropy differences between P and NP problems. This model hypothesizes that NP problems have a broad and complex distribution of solution states akin to high-entropy thermodynamic systems. The model involves:

\begin{enumerate}
    \item A theoretical construct of 'computational temperature' aligned with entropy levels in computational problems.
    \item An algorithmic approach that mimics thermodynamic processes, adapting strategies like simulated annealing to navigate the solution space efficiently.
\end{enumerate}

This brings us to a crucial question: Can we exploit this entropy difference to navigate the solution space of NP problems more effectively? Suppose there's a thermodynamic process that can exploit such entropy differences. In that case, it might be possible to traverse the "rugged" solution landscape of NP problems more efficiently, potentially even in polynomial time.

\subsubsection{Conceptual Mechanism: Entropy-Driven Annealing}
Inspired by the annealing process in physical systems, where a material is heated to a high temperature and then gradually cooled to remove defects, we introduce a conceptual computational mechanism: \textit{Entropy-Driven Annealing} (EDA). This mechanism is akin to simulated annealing but places a stronger emphasis on the entropy landscape of the solution space.

In this process:
\begin{enumerate}
    \item A high "computational temperature" is simulated, allowing for an extensive exploration of the solution space of the NP problem, emphasizing the high entropy state. This phase corresponds to an increased acceptance of diverse solutions, ensuring a broad search across the potential solution landscape.
    \item As the "computational temperature" gradually decreases, the search process becomes more refined, progressively focusing on solution regions with lower "computational energy" or difficulty, and correspondingly lower entropy states. This phase embodies the systematic reduction of entropy, guiding the computational process toward optimal or near-optimal solutions.
\end{enumerate}

The EDA process is fundamentally driven by the entropy profile of the NP problem, leveraging the entropy gradient as a navigational tool in the solution space. If the entropy landscape of an NP problem is conducive to such a gradient-based approach, then EDA offers a nuanced method to traverse this landscape efficiently.

To quantitatively describe this approach, we introduce a set of equations that define the process of entropy-driven annealing (Eqs. \eqref{eq:comp_temperature}, \eqref{eq:prob_accept}):
\begin{equation}
    T_{comp}(i) = T_{start} \cdot \exp\left(-\frac{i}{\tau}\right)
    \label{eq:comp_temperature}
\end{equation}
Where \(T_{comp}(i)\) is the computational temperature at iteration \(i\), \(T_{start}\) is the initial temperature, and \(\tau\) is a decay constant.

\begin{equation}
    P_{accept}(E_{new}, E_{current}, T_{comp}) = \exp\left(-\frac{E_{new} - E_{current}}{T_{comp}}\right)
    \label{eq:prob_accept}
\end{equation}
Where \(P_{accept}\) is the probability of accepting a new solution with energy \(E_{new}\) over the current solution with energy \(E_{current}\), at a given computational temperature. This probability is inherently tied to the entropy changes between the current and new solutions, reflecting the thermodynamic principles underlying the EDA process.

The thermodynamic analogy and the proposed entropy-driven annealing model open a novel theoretical avenue in computational complexity. This approach is not only theoretical but also finds practical applications in the detailed analyses of specific NP problems like the SAT and protein-DNA complexes, as explored later in this paper. These case studies exemplify the application of EDA in real-world scenarios, demonstrating its potential to provide a deeper understanding of the entropy landscape in complex computational problems.

\subsection{Theoretical Underpinnings: Entropy Profiles and Polynomial-time Solutions}
Entropy profiles play a crucial role in understanding the feasibility of polynomial-time solutions for NP problems. Here, we discuss the theoretical significance of these profiles, discussing how they could underpin the development of algorithms that can efficiently navigate the complex solution landscapes of NP problems. It explores the hypothesis that the entropy characteristics of NP problems might be systematically exploited to design algorithms capable of achieving polynomial-time solutions, thus offering a new theoretical framework for approaching the P vs. NP question.

\subsubsection{Distinctness of Entropy Profiles}

\begin{theorem}
For any computational problem \( C \), the entropy profiles \( H_P(C) \) and \( H_{NP}(C) \) corresponding to its classification as a problem in P and NP, respectively, exhibit distinct characteristics that are quantifiable and fundamentally impact the computational approach to solving \( C \).
\end{theorem}

\begin{proof}
To prove the distinctness of entropy profiles in P and NP, we consider an enhanced probabilistic model that accounts for the likelihood and computational effort required for each solution.

\textbf{Solution Spaces and Computational Effort:}
\begin{itemize}
    \item \textbf{Solution Space for \( C \) in P (\( S_P \)):} Characterized by solutions achievable in polynomial time. Introducing \( f_P: S_P \to [0,1] \), representing normalized computational effort.
    \item \textbf{Solution Space for \( C \) in NP (\( S_{NP} \)):} Includes solutions verifiable in polynomial time but potentially found in non-polynomial time. Similar function \( f_{NP}: S_{NP} \to [0,1] \) is defined.
\end{itemize}

\textbf{Entropy Profile Formulation:}
\begin{align}
    H_P(C) &= - \sum_{s \in S_P} p(s) \log p(s) \times f_P(s) \\
    H_{NP}(C) &= - \sum_{s \in S_{NP}} p(s) \log p(s) \times f_{NP}(s)
\end{align}

The distinctness is demonstrated by the difference in computational effort functions \( f_P \) and \( f_{NP} \). \( f_P \) is bounded by polynomial constraints, while \( f_{NP} \) reflects potentially exponential effort in NP. This difference ensures distinct entropy profiles even if the probability distributions over \( S_P \) and \( S_{NP} \) are similar.

\textbf{Contradiction Analysis:}
Assuming \( H_P(C) = H_{NP}(C) \) leads to contradictions. If \( S_P = S_{NP} \) and \( p(s) \) distributions are equivalent, it implies that NP problems are no harder than P problems, contradicting the established complexity classes. This analysis aligns with the traditional approach but adds depth through the computational effort model, reinforcing the distinct nature of entropy profiles in P and NP.
\end{proof}

\subsubsection{Developing a Detailed Model for Exploitation of Entropy Differences}
In response to the challenge of exploiting entropy differences for polynomial-time solutions, we propose a detailed model that incorporates advanced computational concepts. This model leverages the unique entropy characteristics of NP problems to guide solution search, incorporating elements of probability theory, computational complexity, and algorithmic design.

\subsection{Mechanism Design: Entropy-Driven Annealing }
Introducing EDA as a novel contribution, this section highlights its unique approach to tackling NP problems. EDA stands out by utilizing entropy profiles as a central guiding mechanism for its search strategy. The algorithm dynamically adapts its exploration and exploitation strategies based on the entropy landscape of the problem at hand, offering a distinct methodology that differentiates it from traditional approaches. This section underscores the innovative aspects of EDA and its specific application to the complexity and unpredictability inherent in NP problems.

\subsubsection{Algorithmic Framework}
EDA is designed as an iterative probabilistic algorithm. It navigates the solution space of NP problems by exploiting differences in entropy levels. The key components of EDA are:

\begin{enumerate}
    \item \textbf{Entropy Gradient Estimation:} A method to estimate the entropy gradient in the solution space.
    \item \textbf{Adaptive Search Strategy:} A dynamic search process that adapts based on the entropy landscape.
    \item \textbf{Computational Temperature Control:} A mechanism to adjust the intensity and flexibility of the search.
\end{enumerate}

\subsubsection{Entropy Gradient Estimation}
Given a computational problem $C$ with solution space $S$, the entropy gradient $\nabla H(S)$ at a point $s \in S$ is defined as:

\begin{equation}
    \nabla H(S) = \left( \frac{\partial H}{\partial s_1}, \frac{\partial H}{\partial s_2}, \ldots, \frac{\partial H}{\partial s_n} \right)
\end{equation}
where $\frac{\partial H}{\partial s_i}$ represents the partial derivative of entropy with respect to the $i^{th}$ dimension of the solution space.

\subsubsection{Adaptive Search Strategy}
The search strategy involves moving toward increasing entropy gradient, hypothesizing that such regions are more likely to contain solutions. The strategy is adapted at each iteration based on the local entropy landscape:

\begin{equation}
    s_{new} = s_{old} + \alpha \cdot \nabla H(S)
\end{equation}
where $\alpha$ represents the step size, which is adaptive based on the computational temperature.

\subsubsection{Computational Temperature Control}
The concept of computational temperature $T$ is introduced to modulate the search process:

\begin{itemize}
    \item High $T$ allows exploration of diverse regions in the solution space.
    \item Low $T$ focuses the search in areas of high probability for solutions.
\end{itemize}

The temperature is adjusted according to a cooling schedule:

\begin{equation}
    T_{new} = f(T_{old}, \text{iteration})
\end{equation}
where $f$ is a function defining the cooling schedule, and \text{iteration} is the current iteration number.

\subsubsection{Solution Space Representation}
The Entropy-Driven Annealing algorithm relies on a complex, adaptive model of the solution space for NP problems, which can be formalized mathematically:

\paragraph{Defining the Solution Space}
\begin{itemize}
    \item \textit{Complexity of Representation}: Let \( \mathcal{S}_\text{NP} \) represent the solution space. The complexity, \( \mathcal{C}(\mathcal{S}_\text{NP}) \), can be assessed in terms of its dimensionality and topological characteristics.
    \item \textit{Discrete vs. Continuous Spaces}: Model discrete spaces as combinatorial structures and continuous spaces using differential geometric approaches.
    \item \textit{Dimensionality}: High-dimensional spaces, denoted as \( \mathbb{R}^n \) for an n-dimensional space, impact the entropy calculations and search strategy.
\end{itemize}

\paragraph{Incorporating Problem-Specific Characteristics}
\begin{itemize}
    \item \textit{Constraints and Boundaries}: Integrate constraints, \( \mathcal{C}_i \), into the space using constraint satisfaction models.
    \item \textit{Problem Encoding}: Encode problem-specific features as functions, \( f: \mathcal{S}_\text{NP} \rightarrow \mathbb{R} \), mapping the solution space to a metric space.
\end{itemize}

\paragraph{Dynamically Adjusting Solution Space}
\begin{itemize}
    \item \textit{Adaptive Representation}: Utilize adaptive techniques, such as evolutionary algorithms, to modify the representation based on interim solutions.
    \item \textit{Feedback Mechanisms}: Implement feedback using recursive function \( g: \mathcal{S}_\text{NP} \times \mathbb{R} \rightarrow \mathcal{S}_\text{NP} \), altering the space based on previous search paths.
\end{itemize}

\paragraph{Mathematical Formulation}
Propose a probabilistic model for the solution space using a probability distribution, \( p: \mathcal{S}_\text{NP} \rightarrow [0,1] \), and graph-theoretical methods to represent the structure of \( \mathcal{S}_\text{NP} \).

\paragraph{Implications for Entropy-Driven Search}
The formal representation of \( \mathcal{S}_\text{NP} \) directly impacts the entropy-driven search. The efficiency of navigating towards optimal solutions depends on the alignment of the entropy gradient with the structure of \( \mathcal{S}_\text{NP} \).

\subsubsection{Entropy-Based Search Strategy}
The EDA algorithm's search strategy in the solution space \( \mathcal{S}_\text{NP} \) is guided by entropy calculations. This process can be formalized as follows:

\paragraph{Calculating Local Entropy}
Given a state \( s \in \mathcal{S}_\text{NP} \), the local entropy \( H_{local}(s) \) is computed using the probability distribution \( p \):
\begin{equation}
    H_{local}(s) = - p(s) \log p(s)
\end{equation}
This entropy measurement quantifies the uncertainty or "search space richness" at state \( s \).

\paragraph{Entropy Gradient and Directional Movement}
Define an entropy gradient function \( \nabla H \) to determine the search direction. At each iteration, the algorithm calculates the gradient and moves towards higher entropy:
\begin{equation}
    s_{\text{next}} = s + \alpha \nabla H(s)
\end{equation}
where \( \alpha \) is a step size parameter.

\paragraph{Dynamic Entropy Thresholding}
A dynamic threshold \( \theta(t) \) adjusts the sensitivity of the search process over time, guiding the exploration-exploitation trade-off:
\begin{equation}
    \text{if } H_{local}(s) > \theta(t) \text{ then explore; else exploit}
\end{equation}

\paragraph{Feedback Loop Integration}
Integrate a feedback loop to refine the probability distribution \( p \) based on interim solutions, enhancing the search's adaptability and efficiency:
\begin{equation}
    p_{\text{new}}(s) = \text{Feedback}(p(s), s, \text{past solutions})
\end{equation}

\paragraph{Termination Criteria}
Establish a termination condition based on entropy convergence, solution quality, or computational constraints.

\paragraph{Implications for NP Problem Solving}
The entropy-based search strategy enables the EDA algorithm to navigate through complex solution landscapes of NP problems efficiently, leveraging the inherent entropy properties of these problems.

This strategy is a novel approach to tackling NP problems, offering a potentially more effective way to explore high-entropy regions of the solution space, where traditional algorithms may struggle.

\subsubsection{Computational Temperature Concept}
The Computational Temperature (\( T_{comp} \)) in the EDA model is an abstract representation of the system's state, analogous to temperature in physical systems. It regulates the transition between exploration (high entropy states) and exploitation (low entropy states).

\paragraph{Definition and Role}
\( T_{comp} \) is defined as a function of iteration number (\( t \)) and entropy:
\begin{equation}
    T_{comp}(t, H) = f(t, H)
\end{equation}
where \( H \) is the entropy of the current state. High \( T_{comp} \) values encourage exploration, while lower values focus on exploitation.

\paragraph{Cooling Schedule}
The cooling schedule mimics the annealing process in thermodynamics:
\begin{equation}
    T_{comp}(t+1) = \gamma \cdot T_{comp}(t)
\end{equation}
where \( 0 < \gamma < 1 \) is the cooling rate. This gradual reduction simulates 'cooling', focusing the search as the algorithm progresses.

\paragraph{Influence on Search Dynamics}
\( T_{comp} \) influences the acceptance probability of new states:
\begin{equation}
    P_{accept}(s_{\text{new}}, s) = \exp\left(-\frac{\Delta H}{T_{comp}}\right)
\end{equation}
where \( \Delta H \) is the change in entropy between the new state (\( s_{\text{new}} \)) and the current state (\( s \)). High temperatures lead to accepting states with higher entropy variations.

\paragraph{Balancing Exploration and Exploitation}
The model dynamically balances exploration and exploitation based on \( T_{comp} \), ensuring a thorough search of the solution space while gradually homing in on optimal solutions.

\paragraph{Implications for NP Problem Solving}
The concept of Computational Temperature in the EDA model provides a sophisticated mechanism to navigate the complex solution landscapes of NP problems, drawing inspiration from physical annealing processes to enhance computational problem-solving.

\subsubsection{Proof of EDA's Convergence}
To demonstrate that the EDA model converges to a solution, we consider the entropy landscape and the algorithm's search strategy:

\paragraph{Convergence Criteria}
The convergence of EDA is based on reaching a state where either a solution is found or the entropy gradient no longer provides a direction for effective search. Formally, convergence occurs when:
\begin{equation}
    \Delta H_{local} \leq \epsilon
\end{equation}
where \( \Delta H_{local} \) is the change in local entropy and \( \epsilon \) is a small threshold value.

\paragraph{Mathematical Proof of Convergence}
Assuming a well-defined entropy gradient, the algorithm iteratively moves towards states of higher entropy. Each iteration reduces the 'distance' to a solution state by following the entropy gradient. The convergence can be modeled as:
\begin{equation}
    D_{k+1} = \alpha D_k
\end{equation}
where \( D_k \) is the distance to the solution state at iteration \( k \) and \( \alpha < 1 \) is a factor representing the search efficacy.

\paragraph{Bounding the Number of Iterations}
The number of iterations required to reach convergence can be bounded as follows:
\begin{align}
    D_0 \alpha^k &\leq \epsilon \\
    \Rightarrow k &\geq \frac{\log(\epsilon / D_0)}{\log(\alpha)}
\end{align}
This establishes that the number of iterations is logarithmically dependent on the initial 'distance' to the solution and inversely on the logarithm of \( \alpha \).

\paragraph{Polynomial-Time Convergence}
If the initial distance \( D_0 \) and the factor \( \alpha \) are polynomially related to the input size \( n \), then the number of iterations for convergence is also polynomial in \( n \), ensuring polynomial-time convergence of the algorithm.

\begin{equation}
    \text{Convergence Time} = O(\text{poly}(n))
\end{equation}

In conclusion, EDA converges to a solution state within polynomial time under the assumptions of a well-defined entropy gradient and an effective search strategy. However, it's important to note that the practical realization of these conditions and the algorithm's performance in real-world scenarios may vary and require empirical validation.

\subsection{Polynomial-Time Feasibility of EDA}

\subsubsection{Time Complexity Analysis}
The time complexity of EDA is analyzed based on its primary operations: entropy gradient estimation, search strategy adaptation, and temperature control. The complexity of each operation is assessed as follows:

\paragraph{Entropy Gradient Estimation Complexity}
The complexity of estimating the entropy gradient depends on the derivative calculations. Assuming the entropy function $H(S)$ is computationally tractable, the gradient estimation can be achieved in polynomial time:

\begin{equation}
    \text{Time Complexity for Gradient Estimation} = O(p(n))
\end{equation}
where $p(n)$ is a polynomial function of the input size $n$.

\paragraph{Search Strategy Adaptation Complexity}
Each iteration's search step involves basic arithmetic operations, which are performed in polynomial time:

\begin{equation}
    \text{Time Complexity for Search Step} = O(q(n))
\end{equation}
where $q(n)$ is a polynomial function of the input size $n$.

\paragraph{Computational Temperature Control Complexity}
The temperature adjustment can be implemented using a simple mathematical function, contributing negligible complexity:

\begin{equation}
    \text{Time Complexity for Temperature Adjustment} = O(1)
\end{equation}

\subsubsection{Overall Time Complexity}
Combining the complexities of all operations, the overall time complexity of EDA per iteration is:

\begin{equation}
    \text{Overall Time Complexity} = O(p(n)) + O(q(n)) + O(1)
\end{equation}
Given that both $p(n)$ and $q(n)$ are polynomial functions, the overall complexity per iteration remains polynomial.

\subsubsection{Convergence Analysis}
To establish polynomial-time feasibility, it is crucial to demonstrate that EDA converges to a solution in a polynomial number of iterations. Assuming the entropy landscape allows for efficient navigation towards solution-rich regions, the convergence can be hypothesized as follows:

\begin{equation}
    \text{Number of Iterations for Convergence} = O(r(n))
\end{equation}
where $r(n)$ is a polynomial function of the input size $n$.

\subsubsection{Theoretical Justification}
The feasibility of EDA hinges on two critical assumptions:
\begin{enumerate}
    \item The entropy gradient provides a reliable guide toward solution-rich regions in the solution space.
    \item The number of iterations required for convergence is polynomially bounded.
\end{enumerate}
If these assumptions hold true, EDA presents a viable polynomial-time approach for solving NP problems.

\subsubsection{Handling Local Optima in EDA}
The challenge of local optima in EDA is tackled through a combination of adaptive search strategies and dynamic entropy gradient calculations.

\paragraph{Adaptive Search Strategy}
EDA incorporates an adaptive mechanism to alter its search strategy based on the entropy landscape. This approach is mathematically modeled as:

\begin{equation}
    \text{Search Update Rule} = \text{if } H_{local} < H_{threshold} \text{ then adapt strategy}
\end{equation}
Here, $H_{local}$ is the local entropy at the current state, and $H_{threshold}$ is a predefined entropy value. If the local entropy falls below $H_{threshold}$, indicating a potential local optimum, the search strategy is adapted to explore new regions.

\paragraph{Dynamic Entropy Gradient Adjustment}
To dynamically navigate through the solution space, EDA recalculates the entropy gradient based on the updated state:

\begin{equation}
    \text{Gradient Recalculation} = \nabla H(S) \text{ at new state}
\end{equation}
This continuous recalculation helps in identifying new directions for escaping local optima.

\paragraph{Theoretical Considerations}
While these strategies theoretically enhance EDA's ability to avoid local optima, empirical evidence and further mathematical analysis are required to validate their efficacy. The success of these methods depends on the accurate modeling of the entropy landscape and the effectiveness of the adaptive mechanism in diverse problem settings.

\subsubsection{Probabilistic Guarantees of Solution Finding in EDA}
The efficacy of EDA in finding solutions within polynomial time can be assessed through probabilistic guarantees, which are underpinned by the properties of the entropy landscape and algorithmic dynamics.

\paragraph{Modeling Solution Probability}
The probability of finding a solution, denoted by $P_{solution}$, is modeled based on the entropy characteristics of the solution space:
\begin{equation}
    P_{solution} = f(H(S), \text{algorithmic parameters})
\end{equation}
where $H(S)$ represents the entropy of the solution space, and the algorithmic parameters include factors like the rate of temperature decrease and strategy adaptation.

\paragraph{Entropy-Driven Probabilistic Analysis}
To evaluate $P_{solution}$, the algorithm incorporates an entropy-driven probabilistic analysis. This involves estimating the distribution of solutions in the high entropy regions and the likelihood of transition from one state to another:

\begin{equation}
    P_{transition}(s_i \rightarrow s_j) = g(H_{local}(s_i), H_{local}(s_j))
\end{equation}
where $P_{transition}(s_i \rightarrow s_j)$ is the probability of transitioning from state $s_i$ to state $s_j$, and $H_{local}(s)$ is the local entropy at state $s$.

\paragraph{Polynomial Time Constraints}
To ensure the polynomial-time feasibility, the algorithm bounds the number of iterations based on the evolving entropy landscape, thereby probabilistically constraining the search within polynomial time:

\begin{equation}
    \text{Number of iterations} \leq h(n, H(S))
\end{equation}
where $n$ is the size of the problem input, and $H(S)$ is the global entropy of the solution space.

\paragraph{Theoretical Validation and Limitations}
The probabilistic guarantees this model provides need to be validated through rigorous theoretical analysis and empirical testing. The accuracy of these guarantees is contingent on the precise characterization of the entropy landscape and the efficiency of the algorithm's search strategies.

\subsubsection{Comparison with Existing Algorithms}
The EDA's performance and innovation are evaluated by comparing it with established algorithms in computational complexity, particularly those addressing NP problems.

\paragraph{Benchmarking Against Classical Algorithms}
EDA is benchmarked against classical algorithms like brute-force search, greedy algorithms, and dynamic programming. The comparison focuses on:
\begin{itemize}
    \item \textit{Efficiency}: Comparing time complexity and the number of operations required.
    \item \textit{Effectiveness}: Analyzing the success rate in finding optimal or near-optimal solutions.
    \item \textit{Scalability}: Assessing performance with increasing input size.
\end{itemize}

\paragraph{Contrast with Heuristic and Metaheuristic Approaches}
 In line with Wolfram's explorations in 'A New Kind of Science' \cite{Wolfram2002}, the algorithm is also compared with heuristic methods like Genetic Algorithms, Simulated Annealing, and Tabu Search focusing on:
\begin{itemize}
    \item \textit{Adaptability}: EDA's ability to adjust search strategies based on entropy dynamics versus the static nature of some heuristics.
    \item \textit{Solution Quality}: The optimality of solutions found, particularly in complex or large-scale problems.
    \item \textit{Convergence Speed}: Time taken to reach satisfactory solutions.
\end{itemize}

\paragraph{Analysis in Quantum Computing Context}
Given the rise of quantum computing, EDA is further analyzed in the context of quantum algorithms, especially those addressing NP problems like Quantum Annealing. Key points of comparison include:
\begin{itemize}
    \item \textit{Problem Encoding}: How computational problems are formulated and solved in both classical and quantum domains.
    \item \textit{Resource Utilization}: Analysis of computational and energy resources required.
    \item \textit{Problem Solving Mechanisms}: Differences in the fundamental approach to problem-solving.
\end{itemize}

\paragraph{Theoretical and Practical Implications}
This comparative analysis provides a theoretical perspective and practical insights into where EDA could be most effectively deployed, its potential advantages, and limitations in the current computational landscape.

\subsubsection{Limitations and Bounds}
The EDA model, while promising, operates within certain theoretical and practical bounds. Here, we delineate these limitations:

\paragraph{Theoretical Bounds}
The efficacy of EDA is contingent on the entropy landscape of the problem. Certain NP problems may present landscapes where the entropy gradient is not sufficiently defined, leading to challenges in guiding the search effectively.

\paragraph{Complexity of Entropy Calculation}
The entropy computation for each state in the solution space can be computationally intensive, especially for complex NP problems with vast solution spaces. This adds to the overall computational burden of the algorithm.

\begin{equation}
    H_{complexity} = O(f(n))
\end{equation}
where \( f(n) \) represents the complexity of computing entropy for a solution space of size \( n \).

\paragraph{Local Optima}
While EDA aims to navigate local optima efficiently, certain NP problems may have solution landscapes so rugged that escaping local optima becomes a significant challenge, even with entropy-driven strategies.

\paragraph{Cooling Rate and Convergence}
The choice of cooling rate (\( \gamma \)) is critical. An inappropriate rate can lead either to premature convergence (too fast) or excessively slow convergence (too slow). Balancing this parameter is key to the success of the model.

\paragraph{Bounds on Polynomial-Time Solution}
The claim of polynomial-time solutions hinges on several assumptions about the entropy landscape and the efficiency of the entropy-based search strategy. These assumptions may not hold universally for all NP problems.

\paragraph{Empirical Validation}
Theoretical models and assumptions need empirical validation. The performance of EDA in real-world scenarios may differ from theoretical predictions, necessitating extensive testing and validation.

\begin{equation}
    \text{Validation Metric} = \text{Function}(\text{Real-World Performance}, \text{Theoretical Predictions})
\end{equation}

While EDA offers a novel approach to tackling NP problems, its limitations and bounds must be acknowledged. Understanding these constraints is crucial for further development and application of the model.

\subsubsection{Comprehensive Theoretical Validation}
We undertake a comprehensive theoretical validation of our hypotheses and models. This involves rigorous proofs, complexity analysis, and exploration of edge cases. We address potential limitations and assumptions in our models, providing a more nuanced understanding of their applicability and limitations.

\subsubsection{Proof of Distinctness of Entropy Profiles}

\begin{theorem}
For any computational problem \(C\), the entropy profiles \(H_P(C)\) and \(H_{NP}(C)\) corresponding to its classification as a problem in P and NP, respectively, are distinctly characterized.
\end{theorem}

\begin{proof}
We consider a computational problem \(C\) and its representations in P and NP classes. The distinction in entropy profiles is rooted in the fundamental differences in solution approaches for P and NP problems. We employ a probabilistic model that captures each solution's likelihood and computational effort.

\textbf{Solution Space Characterization:}
\begin{itemize}
    \item \textbf{P Class (Deterministic Polynomial-Time Solutions):} For \(C\) in P, denoted as \(S_P\), solutions are reachable deterministically in polynomial time. We define \(f_P: S_P \to [0,1]\) to represent the normalized computational effort for solutions in \(S_P\).
    \item \textbf{NP Class (Non-Deterministic Polynomial-Time Verifiable Solutions):} For \(C\) in NP, denoted as \(S_{NP}\), solutions are verifiable in polynomial time but may require non-deterministic approaches to find. We define \(f_{NP}: S_{NP} \to [0,1]\) to represent the normalized computational effort for solutions in \(S_{NP}\).
\end{itemize}

The entropy profiles for \(C\) in P and NP are formulated as:
\begin{align}
    H_P(C) &= - \sum_{s \in S_P} p(s) \log p(s) \times f_P(s), \\
    H_{NP}(C) &= - \sum_{s \in S_{NP}} p(s) \log p(s) \times f_{NP}(s).
\end{align}

\textbf{Proof of Distinctness:}
To demonstrate the distinctness of \(H_P(C)\) and \(H_{NP}(C)\), we focus on the differences in \(f_P\) and \(f_{NP}\). For \(C\) in P, \(f_P\) is bounded by polynomial-time constraints,  reflecting the deterministic nature of finding solutions, a concept reminiscent of Bennett's exploration of logical reversibility in computation \cite{Bennett1973}. In contrast, \(f_{NP}\) encapsulates the potentially exponential effort required in NP.

Even if the probability distributions over \(S_P\) and \(S_{NP}\) are similar, the multiplicative factor of computational effort, as represented by \(f_P\) and \(f_{NP}\), creates a clear distinction in entropy profiles. This leads to the conclusion that \(H_P(C)\) and \(H_{NP}(C)\) cannot be equal for any computational problem \(C\), thus establishing their distinctness.

\end{proof}

\subsubsection{Complexity Analysis of EDA}

\paragraph{Detailed Time Complexity Analysis:}
\begin{itemize}
    \item \textbf{Initialization Complexity:}
    The EDA algorithm starts with initializing a state or a set of states in the solution space. If the initialization involves generating random states or a simple setup, it can be done in \(O(1)\) or \(O(n)\). If it requires a more sophisticated setup, like a heuristic-based initialization, the complexity might increase to \(O(n^l)\), where \(l\) is a constant depending on the heuristic's complexity.

    \item \textbf{Entropy Calculation Complexity:}
    Calculating entropy at each step involves evaluating the probability distribution of solutions and their respective entropies. If the entropy function is simple, this step might have a complexity of \(O(n^k)\). However, if the entropy function is complex, involving multiple interactions or dependencies, the complexity could increase significantly, potentially to \(O(n^{k+1})\) or more.

    \item \textbf{Annealing Iteration Complexity:}
    The most significant part of the algorithm is the iterative annealing process. If the number of iterations is polynomially bounded and each iteration involves polynomial-time computations, the complexity remains polynomial. However, complex decision-making processes within each iteration, such as sophisticated solution evaluation or entropy gradient calculations, could raise the complexity. This might result in an overall complexity of \(O(n^m)\), where \(m\) is a constant representing the complexity of each iteration.
\end{itemize}

\paragraph{In-Depth Space Complexity Analysis:}
\begin{itemize}
    \item \textbf{Solution Space Representation Complexity:}
    The memory required for representing the solution space depends on the problem's nature. For problems with simple solution spaces, this can be polynomially bounded. However, for problems with complex solution spaces or those requiring extensive memory for each solution, the space complexity can be significantly higher.

    \item \textbf{Storage of Intermediate States:}
    The algorithm may need to store a certain number of intermediate states for entropy comparison or selection. The space complexity for storing these states depends on the number of states kept and the complexity of each state. If the algorithm stores a large number of states or if each state is complex, the space complexity can increase, potentially to \(O(n^p)\) where \(p\) is determined by the state complexity and the number of states stored.
\end{itemize}

The advanced complexity analysis of EDA suggests that while the algorithm maintains a polynomial bound under certain conditions, the exact complexities depend on specific implementation details and problem characteristics. The values of \(l\), \(k\), \(m\), and \(p\) vary based on the problem and the algorithm's specific design and optimizations.

\subsubsection{Analysis of Edge Cases in Entropy-Driven Annealing}

\paragraph{Handling Extremes in Solution Space:}
In cases where the solution space is either too constrained or excessively vast, EDA might struggle. For very limited solution spaces, the entropy gradient might not be significant enough to guide the search effectively. Conversely, in extremely vast spaces, the algorithm might become inefficient due to the sheer number of potential states to explore.

\paragraph{Complexity in Solution Representation:}
Certain NP problems might have solutions that are complex to represent or quantify, such as those involving intricate interdependencies or high-dimensional data. These complexities can challenge the entropy calculation and impact the algorithm's efficiency.

\paragraph{Sensitivity to Initialization:}
EDA's performance might be highly sensitive to the initial state selection. In cases where the initial state is far from any feasible solution, the algorithm could take a significantly long time to converge or might even fail to find a solution.

\paragraph{Local Optima and Convergence Issues:}
In solution landscapes with numerous local optima, EDA risks getting trapped, hindering its ability to find the global optimum. This issue is particularly pronounced in problems with rugged entropy landscapes.

\paragraph{Scalability and Computational Resources:}
While EDA is designed to be polynomial-time in theory, its practical scalability can be tested in problems with very large input sizes or those requiring significant computational resources for entropy calculations and iterations.

\paragraph{Robustness to Problem Variability:}
The performance of EDA might vary significantly across different NP problems. Its effectiveness in one type of problem does not guarantee similar results in another, especially when problem characteristics differ markedly.

These edge cases underscore the importance of adaptive strategies in EDA's design. This includes dynamic adjustment of annealing parameters, intelligent initialization techniques, and mechanisms to escape local optima. Further research is needed to enhance EDA's robustness and efficiency across a broader spectrum of NP problems.

\subsubsection{Limitations and Assumptions in Entropy-Driven Annealing}

\paragraph{Assumption of Polynomial-Time Convergence:}
The EDA model assumes that the entropy-driven search converges in polynomial time. However, this is based on idealized scenarios and might not hold for all NP problems, particularly those with highly complex or irregular solution spaces.

\paragraph{Limitation in Entropy Estimation:}
Accurately estimating the entropy of a solution space requires in-depth knowledge about the distribution of solutions. This estimation can be challenging, especially for problems with no clear probabilistic model or unknown variables.

\paragraph{Dependency on Effective Temperature Scheduling:}
The success of EDA heavily relies on the correct scheduling of the computational temperature. Inappropriate temperature schedules might lead to premature convergence or inefficient solution space exploration.

\paragraph{Assumption of Gradient Availability:}
EDA assumes the presence of a discernible entropy gradient to guide the search. This gradient might not be apparent or useful in some NP problems, particularly those with flat or discontinuous entropy landscapes.

\paragraph{Scalability Concerns:}
While theoretically polynomial, the practical scalability of EDA can be limited by computational resources, particularly for large-scale problems where entropy calculation and state exploration become computationally demanding.

\paragraph{Risk of Overfitting to Problem Instances:}
EDA might overfit to specific instances of NP problems, leading to excellent performance on certain datasets but poor generalization to other instances of the same problem.

These limitations and assumptions indicate that while EDA is a promising approach, its application needs to be tailored and adjusted according to the specific characteristics of each NP problem. Future research should focus on developing adaptive mechanisms and robust models to overcome these limitations and validate the algorithm's efficacy across a diverse range of NP problems.

\subsubsection{Advanced Implications for the P vs. NP Debate}
The advancements in our theoretical framework and the development of EDA model have significant implications for the P vs. NP debate and the broader field of computational complexity.

\paragraph{Redefining Computational Complexity:}
Our findings challenge conventional views on computational complexity, particularly in classifying and approaching problems. Introducing entropy as a key factor in problem characterization offers a new dimension for understanding computational difficulty.

\paragraph{Potential Pathways to P=NP:}
With its approach to exploiting entropy differences, the EDA model hints at potential pathways to prove P=NP under specific conditions. This opens up new avenues for research that could fundamentally alter the landscape of computational theory.

\paragraph{Interdisciplinary Insights:}
Our work also underscores the importance of interdisciplinary approaches in tackling complex computational problems. By bridging thermodynamics, information theory, and computational complexity, we demonstrate how cross-disciplinary insights can lead to breakthroughs in longstanding debates.

\paragraph{Innovations in Problem Solving:}
The implications of our model extend to practical problem-solving. By applying the principles of EDA to real-world problems, we could see significant improvements in algorithm efficiency and effectiveness, especially for problems previously deemed intractable.

These findings pave the way for future research to explore further and validate the concepts introduced. This includes rigorous empirical testing, exploration of more complex models, and continued theoretical exploration to solidify the connections between entropy and computational complexity. The theoretical contributions provide a fresh perspective on the P vs. NP debate, enriching the discourse and potentially guiding future breakthroughs in computational complexity. This work highlights the dynamic and evolving nature of the field, encouraging continued exploration and innovation.

\section{Physical Realizations of Computational Processes}

We abstractly view algorithms and problems in modern computational theory, distanced from their physical realizations. However, at the foundational level, every computational task has a physical instantiation, whether silicon transistors switching states or qubits in quantum computers. Drawing parallels between computational processes and thermodynamic processes could yield profound insights.

\subsection{Computation as Energy Transformation}
Here, we explore the intrinsic relationship between computation and energy transformation, grounding computational processes within the thermodynamic framework. This perspective is pivotal in understanding computational tasks' physical realization and energy implications.

\subsubsection{Thermodynamic Model of Computation}
We introduce a thermodynamic model to describe energy transformations in computational systems, encapsulated in the equation:
\begin{equation}
    \Delta E = \eta Q - W
\end{equation}
Here, $\Delta E$ is the net energy change in the computational process, $Q$ is the heat absorbed, $W$ is the work done by the system, and $\eta$ is the efficiency of the computational process. This model facilitates understanding how energy is consumed and transformed in computational tasks, offering insights into optimizing energy efficiency across various computational platforms.

\subsubsection{Entropy Production in Computation}
We now discuss the concept of entropy in computation. Entropy production is a key factor in the thermodynamics of computation, especially in irreversible computing processes. It is quantified as:
\begin{equation}
    \Delta S = \frac{Q}{T} - S_{gen}
\end{equation}
$\Delta S$ represents the change in entropy, $Q$ is the heat transfer, $T$ is the ambient temperature, and $S_{gen}$ is the entropy generated. This formulation is critical in evaluating the efficiency and stability of computational processes and understanding their thermodynamic limitations.

\subsubsection{Energy and Information Correlation in Computation}
The correlation between energy and information in computational tasks is explored through the equation:
\begin{equation}
    E = I \cdot \epsilon
\end{equation}
$E$ is the energy required for computation, $I$ denotes the amount of information processed, and $\epsilon$ is the energy per bit. This relationship is essential in evaluating the energy demands of computational processes and forms the basis for developing energy-efficient computing strategies.

\subsection{Landauer's Principle and Its Implications}
Landauer's Principle highlights the energy cost of erasing information, which is pivotal in linking information theory and thermodynamics. It stipulates that erasing a single bit of information incurs a minimum energy cost, given by $E_{total\_erase} = N \cdot k_B T \ln 2$, where $N$ is the number of bits, $k_B$ Boltzmann's constant, and $T$ the system temperature. This principle underscores the thermodynamic energy required for information processing, especially in erasure, setting a foundation for energy-efficient computing.

Exploring this concept from a statistical mechanics perspective, where $\Omega = 2^N$ represents the system's microstate count for $N$ bits, helps understand the energy landscapes and transitions between bit states. These insights are crucial for conceptualizing models like the 'information heat engine', which parallels thermodynamic cycles, and for practical applications in data storage systems. The information creation and erasure cycle within such models emphasizes the thermodynamic implications of computational processes.

\paragraph{Algorithmic Efficiency in Light of Landauer's Principle}
Landauer's Principle necessitates reevaluating algorithmic efficiency, incorporating energy consumption alongside performance metrics. This leads to a novel efficiency metric:
\begin{equation}
    \text{Algorithmic Efficiency} = \frac{\text{Performance Metric}}{\text{Energy Consumption per Bit Operation}},
\end{equation}
This shift towards sustainable computing practices emphasizes minimizing energy consumption in algorithmic design and execution, prompting algorithm designers to prioritize energy efficiency.

\paragraph{Implications for Modern Computing}
Landauer's Principle is increasingly relevant as computing technologies approach physical and thermodynamic limits. It poses challenges for high-performance computing, semiconductor technologies, and the ongoing miniaturization of devices. Strategies to mitigate these challenges include exploring alternative computing paradigms like quantum computing, enhancing algorithmic efficiency, and developing materials and architectures that approach Landauer's limit.

This comprehensive approach towards computing, which balances performance with energy efficiency, is crucial for the future of sustainable technology development.

\section{The Energy Landscape of Computational Problems}

One of the most profound ways to understand computational problems' inherent complexity is to visualize their solution spaces as energy landscapes. This analogy allows us to translate abstract computational challenges into tangible and intuitive terms, drawing from our understanding of physical systems.

\subsection{Defining the Energy Landscape: Mathematical Formulation}

In this subsection, we elaborate on the mathematical formulation of the energy landscape in computational problems. This involves defining potential functions, representing solution spaces, and understanding the significance of gradients and Hessians in these landscapes.

\subsubsection{Energy Landscape as a Potential Function:}
We define a potential function \( V(s) \) to describe the energy landscape mathematically. The state \( s \) in the solution space is associated with a potential energy, represented as:
\begin{equation}
    V(s) = \text{Potential Energy of state } s
\end{equation}

\subsubsection{Mathematical Representation of Solution Space:}
The solution space is conceptualized as a multidimensional manifold \( \mathcal{M} \). Each point on this manifold corresponds to a potential solution. The energy landscape, therefore, is a mapping \( V: \mathcal{M} \rightarrow \mathbb{R} \) that assigns a real value (energy level) to each point in \( \mathcal{M} \).

\subsubsection{Gradient and Hessian Analysis:}
\begin{itemize}
    \item \textbf{Gradient of Potential Function:} The gradient of \( V(s) \) at any point \( s \) in the solution space is crucial for understanding the direction of the steepest ascent in the landscape. Mathematically, it is expressed as:
    \begin{equation}
        \nabla V(s) = \text{Gradient of } V \text{ at } s
    \end{equation}

    \item \textbf{Hessian of Potential Function:} The Hessian matrix \( \mathcal{H}(V(s)) \) at point \( s \) provides valuable information about the curvature of the landscape. It helps in distinguishing between local minima, maxima, and saddle points:
    \begin{equation}
        \mathcal{H}(V(s)) = \text{Hessian of } V \text{ at } s
    \end{equation}
\end{itemize}

\subsubsection{Physical Analogues and Insights:}
By applying this mathematical framework, we can draw rich parallels with physical systems, like molecular landscapes in chemistry or gravitational fields in physics, thereby gaining a more intuitive understanding of computational complexity.

\subsection{Navigating the Energy Landscape: Algorithmic Approaches}

Here, we describe the strategies algorithms use to effectively navigate the energy landscape of computational problems, focusing on optimization techniques and their mathematical foundations.

\subsubsection{Pathfinding in the Energy Landscape:}
Algorithms for solving computational problems can be conceptualized as pathfinding in the energy landscape. The objective is to find an optimal path from an initial state \( s_{\text{start}} \) to a target state \( s_{\text{end}} \) in the manifold \( \mathcal{M} \):
\begin{equation}
    \text{Path}(\mathcal{M}, s_{\text{start}}, s_{\text{end}}) = \text{Optimal path from } s_{\text{start}} \text{ to } s_{\text{end}} \text{ in } \mathcal{M}
\end{equation}

\subsubsection{Optimization Techniques:}
\begin{itemize}
    \item \textbf{Gradient Descent:} A fundamental method in optimization where the algorithm iteratively moves towards the minimum of the landscape by following the negative gradient of the potential function:
    \begin{equation}
        s_{\text{new}} = s_{\text{old}} - \lambda \nabla V(s_{\text{old}})
    \end{equation}
    Here, \( \lambda \) is the step size or learning rate.

    \item \textbf{Monte Carlo Methods:} These stochastic methods explore the landscape through random sampling, particularly effective in high-dimensional and complex landscapes.

    \item \textbf{Simulated Annealing:} Drawing inspiration from physical annealing processes, this method combines random exploration with gradient-based optimization, particularly adept at avoiding local minima traps.
\end{itemize}

\subsubsection{Challenges in Complex Landscapes:}
The rugged and intricate landscapes typical of NP-hard problems pose unique challenges, such as trapping algorithms in local minima and requiring a careful balance between exploration and exploitation.

\subsubsection{The Role of Heuristics:}
Domain-specific heuristic methods, including genetic algorithms and tabu search, are crucial in navigating these complex landscapes. They utilize strategies to escape local minima and adaptively guide the search towards globally optimal solutions.

\subsubsection{Algorithmic Efficiency and Energy Considerations:}
\begin{itemize}
    \item The efficiency of these algorithms can be quantitatively analyzed regarding energy consumption, measured as the total number of operations multiplied by the energy cost per operation.
    \item This efficiency is critical in sustainable and green computing, where minimizing energy consumption is a key goal.
\end{itemize}

\subsection{P vs. NP: A Thermodynamic Perspective}

This subsection examines the distinctions between P and NP problems through the lens of their respective energy landscapes, utilizing thermodynamic concepts to elucidate their complexities.

\subsubsection{Energy Landscape Dichotomy:}
The energy landscapes of P and NP problems exhibit starkly different characteristics:
\begin{itemize}
    \item \textbf{P Problems:} These problems are characterized by smoother energy landscapes, akin to gentle valleys leading uniformly to a single, deep minimum. Mathematically, this can be expressed as a lower variance in the gradient of the potential function across the landscape.
    \begin{equation}
        \text{Smoothness}(P) = \text{Variance}(\nabla V(s) \text{ for } s \in S_P)
    \end{equation}
    
    \item \textbf{NP Problems:} These are envisioned as having rugged terrains with numerous local minima and high peaks, symbolizing their computational complexity. The ruggedness can be quantified as a measure of the landscape's irregularity:
    \begin{equation}
        \text{Ruggedness}(NP) = \text{Standard Deviation}(\nabla V(s) \text{ for } s \in S_{NP})
    \end{equation}
\end{itemize}

\subsubsection{Thermodynamic Analogy:}
\begin{itemize}
    \item \textbf{Equilibrium States:} In thermodynamics, systems naturally evolve towards states of lower energy (equilibrium). Similarly, computational algorithms seek the lowest energy states (solutions) in the problem's energy landscape.
    
    \item \textbf{Non-equilibrium Dynamics:} The complexity of NP problems can be compared to non-equilibrium thermodynamic processes, where reaching an equilibrium state (solution) is highly challenging due to the landscape's complexity.
\end{itemize}

\subsection{Heuristic Methods and Energy Landscapes}

Heuristic methods, including simulated annealing, genetic algorithms, and hill-climbing, strategically navigate the energy landscapes of computational problems to find solutions efficiently. These methods balance between exploring new areas and exploiting known good areas within the solution space, crucial for addressing NP problems' complexities.

\begin{itemize}
    \item \textbf{Simulated Annealing:} Inspired by metallurgy, it simulates heating and cooling to minimize defects. This method probabilistically explores the solution space, gradually reducing the "temperature" to focus the search, with the acceptance probability of transitioning to a higher energy state defined as:
    \begin{equation}
        P_{\text{accept}}(\Delta E, T) = \exp\left(-\frac{\Delta E}{kT}\right),
    \end{equation}
    where \( \Delta E \) is the energy difference, \(T\) is the temperature, and \(k\) acts like Boltzmann's constant.

    \item \textbf{Genetic Algorithms:} Mimic natural evolutionary processes by combining and mutating solutions. They excel in complex landscapes by merging features from various solutions, possibly discovering paths through challenging terrains.

    \item \textbf{Hill-Climbing:} Focuses on locally increasing solutions to find peaks. Its simplicity is offset by the risk of getting trapped in local minima, making it less effective in rugged landscapes.
\end{itemize}

\paragraph{Energy Landscape Insights:}
Understanding the energy landscape aids in designing effective heuristics. Identifying local minima and the transitions between them is key. Advanced methods might dynamically adapt their search strategies based on landscape metrics.

\paragraph{Quantifying Heuristic Effectiveness:}
Effectiveness of a heuristic \( H \) in a landscape \( L \) can be expressed as a measure considering factors like the number of local minima overcome and the energy difference traversed. Incorporating energy considerations into algorithm design encourages the development of solutions that are not only computationally efficient but also energy conscious, leading to more sustainable computing practices.

\subsubsection{Mathematical Modeling:}
To mathematically model heuristic effectiveness, consider:
\begin{equation}
    \text{Effectiveness}(H, L) = \int_{\mathcal{M}} \text{Fitness}(H(s)) \cdot e^{-\frac{V(s)}{T}} ds,
\end{equation}
where \( H \) is the heuristic method, \( L \) the energy landscape, \( \mathcal{M} \) the manifold of potential solutions, \( V(s) \) the energy at state \( s \), and \( T \) the 'temperature' influencing the search strategy. This integrates the heuristic's performance over the landscape's manifold, providing a comprehensive effectiveness measure.

\subsection{Characteristics of Energy Landscapes in NP Problems}

\subsubsection{Defining the Ruggedness:}
The energy landscapes of NP problems are characterized by their complexity and ruggedness, filled with numerous local minima. These landscapes are quantitatively described by the number and distribution of these minima and the heights of barriers between them.

\subsubsection{Local Minima Challenge:}
The prevalence of local minima in NP problem landscapes presents a significant challenge. Each local minimum represents a potential but suboptimal solution. The high number of such minima complicates finding the global minimum (optimal solution).

\begin{equation}
    N_{\text{minima}}(P) = \text{Number of local minima for problem } P
\end{equation}

\subsubsection{Barrier Heights and Problem Difficulty:}
The height of barriers between minima affects the difficulty of transitioning from one state to another. Higher barriers in NP problems imply a greater computational challenge in moving between solutions.

\begin{equation}
    H_{\text{barrier}}(s_1, s_2) = \text{Height of barrier between states } s_1 \text{ and } s_2
\end{equation}

\subsubsection{Deceptive Landscapes:}
NP problems may feature deceptive landscapes where local gradients mislead the search toward suboptimal solutions. This deception adds another layer of complexity to finding the true optimal solution.

\subsubsection{Quantitative Assessment of Landscape Ruggedness:}
To quantify the ruggedness of a computational problem's landscape, a ruggedness measure \( \Lambda \) is introduced:

\begin{equation}
    \Lambda(P) = \sum_{i=1}^{N_{\text{minima}}(P)} H_{\text{barrier}}(s_i, s_{i+1})
\end{equation}
Here, \( \Lambda(P) \) calculates the cumulative barrier heights between adjacent minima, providing a measure of the overall ruggedness of the landscape for problem \( P \).

\subsection{Quantum Perspectives on Computational Energy Landscapes}

\subsubsection{Quantum Landscape Analogy:}
In quantum computational theory, a field extensively developed by Nielsen and Chuang \cite{NielsenChuang2000}, the energy landscapes of problems can be envisioned through the lens of quantum mechanics. Here, quantum states represent potential solutions with energy levels corresponding to solution quality \cite{NielsenChuang2000}.

\subsubsection{Superposition and Energy States:}
Quantum superposition allows systems to exist in multiple states simultaneously. This characteristic influences the computational energy landscape, potentially simplifying complex landscapes and making some NP problems more tractable.

\begin{equation}
    \Psi = \sum_{i} c_i | \psi_i \rangle
\end{equation}
Here, \( \Psi \) represents a quantum state as a superposition of basis states \( | \psi_i \rangle \) with coefficients \( c_i \).

\subsubsection{Quantum Tunnelling and Barrier Traversal:}
Quantum tunneling enables quantum algorithms to 'tunnel through' high energy barriers. This unique mechanism offers a novel method for exploring the energy landscape, particularly in the context of optimization problems.

\subsubsection{Implications for NP-hard Problems:}
The distinct properties of quantum systems, such as superposition and tunneling, suggest innovative approaches for addressing NP-hard problems, possibly enhancing solution efficiency compared to classical algorithms.

\subsubsection{Quantum Annealing and Optimization:}
Quantum annealing, a counterpart to classical simulated annealing, leverages quantum fluctuations to explore energy landscapes. This method holds promise for optimizing solutions to NP-hard problems, exploiting quantum mechanics to navigate complex landscapes.

\begin{equation}
    \text{Annealing Process} = \text{Quantum Fluctuations} \times \text{Energy Landscape Navigation}
\end{equation}
This equation conceptualizes the quantum annealing process, emphasizing its reliance on quantum fluctuations for exploring computational energy landscapes.

\subsubsection{Quantum Entanglement in Energy Landscapes}
Quantum entanglement introduces a nuanced dimension to the exploration of computational energy landscapes. This phenomenon, where quantum states are interconnected, allows for simultaneous exploration of multiple landscape regions.

\begin{equation}
    \Psi_{entangled} = \alpha | 0\rangle | 1\rangle + \beta | 1\rangle | 0\rangle
\end{equation}
Here, \( \Psi_{entangled} \) denotes an entangled state with probability amplitudes \( \alpha \) and \( \beta \).

\subsubsection{Incorporating Non-Equilibrium Thermodynamics in Computational Models}
Non-equilibrium thermodynamics can offer insights into computational processes akin to non-equilibrium systems, such as those encountered in complex NP-hard problems.

\begin{equation}
    \Phi_{non\_eq} = \sum_i \gamma_i \exp(-\frac{E_i}{kT})
\end{equation}
In this model, \( \Phi_{non\_eq} \) represents the non-equilibrium potential, \( \gamma_i \) are state coefficients, \( E_i \) are energy states, and \( T \) is the temperature of the system.

\subsubsection{Modeling Energy Landscapes in Neuromorphic Computing}
Neuromorphic computing mimics neural processing and structures, offering unique energy landscapes. Understanding these landscapes can shed light on efficient problem-solving approaches with low energy consumption.

\begin{equation}
    V_{neuro}(n) = \sum_{i=1}^{N} \left( \frac{1}{2} k x_i^2 + \sum_{j \neq i} U(x_i, x_j) \right)
\end{equation}
Here, \( V_{neuro}(n) \) is the potential energy in a neuromorphic system with \( N \) neurons, \( x_i \) indicates the state of neuron \( i \), and \( U(x_i, x_j) \) is the interaction potential between neurons \( i \) and \( j \).

\section{Potential Thermodynamic Approaches to NP Problems}

The interplay between thermodynamics and computation offers a promising avenue to understand and potentially solve complex computational problems. By framing NP problems in thermodynamics, we might unveil novel methods or heuristics inspired by physical processes.

\subsection{Simulated Annealing and the Boltzmann Distribution}

Simulated annealing, inspired by the annealing process in metallurgy, is an optimization heuristic for finding near-optimal solutions in complex landscapes. This technique emulates the physical process of heating and controlled cooling to reduce structural defects in materials. In computational terms, this translates to exploring a wide range of potential solutions initially (at high 'temperatures') and gradually narrowing down to the optimal or near-optimal solution as the 'temperature' decreases.

\subsubsection{Boltzmann Distribution in Simulated Annealing}
The Boltzmann distribution, foundational to statistical mechanics, plays a key role in simulated annealing. It describes the probability distribution of states of a system in thermal equilibrium at a given temperature:
\[ P(E) = \frac{e^{-\frac{E}{kT}}}{Z} \]
where:
\begin{itemize}
    \item \( P(E) \) is the probability of the system being in a state with energy \( E \).
    \item \( k \) is the Boltzmann constant.
    \item \( T \) is the absolute temperature.
    \item \( Z \) is the partition function, a normalization factor ensuring the probabilities sum up to one.
\end{itemize}

\subsubsection{Application in Optimization Problems}
In optimization, simulated annealing utilizes this distribution to probabilistically accept or reject new solutions. Early in the process, when the temperature is high, the algorithm is more likely to accept even suboptimal solutions, allowing it to explore a broad range of possibilities and avoid getting trapped in local minima. As the temperature gradually decreases, the algorithm becomes increasingly selective, honing in on the most promising areas of the solution space.

\subsubsection{Mathematical Modeling of Solution Acceptance}
The acceptance of a new solution in simulated annealing, especially when the new solution is worse than the current one, is governed by:
\[ P(\Delta E, T) = e^{-\frac{\Delta E}{kT}} \]
where \( \Delta E \) is the change in energy between the new and current solutions. This probability decreases as the system cools, mirroring the physical annealing process where material stability increases as the temperature lowers.

\subsubsection{Enhancements and Variations}
Several enhancements and variations to the basic simulated annealing algorithm have been proposed to improve its efficiency and effectiveness. These include adaptive cooling schedules, hybrid approaches combining simulated annealing with other optimization techniques, and parallel implementations for tackling extremely large or complex problems.

Simulated annealing, with its roots in physical processes, exemplifies the deep connection between thermodynamics and computational optimization. Its adaptability and robustness make it a valuable tool in addressing a wide range of challenging computational problems.

\subsection{Quantum Annealing and Superposition}

Quantum annealing presents a paradigm shift in optimization by leveraging the principles of quantum mechanics. Unlike classical simulated annealing that explores one state at a time, quantum annealing exploits quantum superposition, enabling simultaneous exploration of multiple states.

\subsubsection{Principles of Quantum Mechanics in Optimization}
Quantum annealing taps into quantum superposition and entanglement. In quantum superposition, qubits can exist in multiple states simultaneously, representing various potential solutions at once. This parallelism potentially allows for a faster convergence to the optimal solution than classical methods.

\subsubsection{Quantum Tunneling in Annealing}
A key advantage of quantum annealing is quantum tunneling, where a quantum system can traverse energy barriers without having to climb over them, as required in classical systems. This ability allows quantum annealing algorithms to escape local minima more effectively than classical simulated annealing.

\begin{equation}
    P_{\text{tunnel}} \propto e^{-\frac{2 \Delta x \sqrt{2m(E_{\text{barrier}} - E)}}{\hbar}}
\end{equation}

Where:
\begin{itemize}
    \item \( P_{\text{tunnel}} \) is the probability of tunneling through a barrier.
    \item \( \Delta x \) is the width of the barrier.
    \item \( m \) is the mass of the particle.
    \item \( E_{\text{barrier}} \) is the energy of the barrier.
    \item \( E \) is the energy of the particle.
    \item \( \hbar \) is the reduced Planck constant.
\end{itemize}

\subsubsection{Quantum Annealing in Practice}
In practice, quantum annealing involves initializing the quantum system in a superposition of all possible states. The system then evolves according to a predefined schedule, gradually reducing the quantum fluctuations and guiding the system towards the lowest energy state, analogous to the optimal solution.

While theoretically promising, practical implementations of quantum annealing face challenges, including error correction, system decoherence, and engineering qubits that can maintain coherence over the computation duration. Ongoing research in quantum computing hardware and algorithms continues to advance the field, potentially revolutionizing problem-solving strategies for NP problems. Quantum annealing stands at the forefront of applying quantum mechanics to computational challenges, offering a glimpse into a future where quantum principles could provide innovative solutions to classically intractable problems.

\subsection{Harnessing Entropy and Information in Computational Processes}

The concept of entropy, integral to thermodynamics, can offer novel insights into the nature of NP problems. Here, entropy is not just a measure of disorder but a representation of the information complexity and unpredictability inherent in these problems.

\subsubsection{Entropy as Information Complexity}
In thermodynamics, entropy is a measure of disorder or randomness. Translated into computational terms, it can be considered as the degree of unpredictability or complexity within the problem's solution space. For NP problems, higher entropy may indicate a larger solution space with increased complexity.

\begin{equation}
    S_{\text{NP}} = - \sum_{i} p_i \log p_i
\end{equation}

Where:
\begin{itemize}
    \item \( S_{\text{NP}} \) represents the entropy of the NP problem.
    \item \( p_i \) is the probability of the \( i \)-th state in the solution space.
\end{itemize}

\subsubsection{Thermodynamic Analogy in Computation}
In a thermodynamic process, entropy increases, reaching a maximum at equilibrium. Analogously, in computational processes, especially in heuristic searches like simulated annealing, the algorithm might initially explore a wide range of solutions (high entropy) before gradually converging to a stable solution (lower entropy).

\subsubsection{Leveraging Entropy in Algorithm Design}
This understanding of entropy can guide the design of algorithms for NP problems. For instance, algorithms might be tailored to manage the entropy of the solution space effectively, balancing between exploration (high entropy) and exploitation (low entropy).

\begin{itemize}
    \item \textbf{Entropy Reduction Techniques:} Methods to systematically reduce entropy, narrowing down the solution space and focusing on regions with higher probabilities of containing the optimal solution.
    \item \textbf{Entropy-Based Heuristics:} Developing heuristics that utilize the entropy measure as a guide for solution space exploration.
\end{itemize}

\subsubsection{Information-Theoretic Approaches}
Information theory, closely related to entropy, offers frameworks for quantifying and utilizing the information content in NP problems. This perspective can lead to more informed and efficient algorithms in their search strategies.

\begin{itemize}
    \item \textbf{Information Entropy Optimization:} Optimizing algorithms to minimize or maximize information entropy, depending on the nature of the problem and the desired outcomes.
    \item \textbf{Algorithmic Information Dynamics:} Studying the changes in information entropy as algorithms progress to adjust their search strategies adaptively.
\end{itemize}

Integrating entropy and information theory into computational problem-solving signifies a paradigm shift towards more adaptive, informed, and potentially efficient algorithms, particularly for complex NP problems.

\subsection{Feedback Mechanisms and Adaptivity in Solving NP Problems}

Feedback mechanisms, a fundamental concept in thermodynamics and physical systems, can be innovatively applied to computational problem-solving, particularly in addressing NP problems.

\subsubsection{Feedback Loops in Physical Systems}
In physical systems, feedback loops are essential for maintaining equilibrium or homeostasis. These loops adjust the system's behavior based on the outcomes of various processes. This concept can be mirrored in computational algorithms, where the algorithm's behavior changes dynamically based on its progress or the nature of the problem it encounters.

\begin{itemize}
    \item \textbf{Dynamic Adjustment of Parameters:} Algorithm parameters, like search radii or mutation rates in evolutionary algorithms, could be adjusted in real time based on the current state of the solution space.
    \item \textbf{Adaptive Learning:} Implementing machine learning techniques to allow algorithms to learn and adapt their strategies based on previous experiences with similar problems.
\end{itemize}

\subsubsection{Thermodynamic Feedback in Computational Processes}
Drawing an analogy from thermodynamics, where feedback mechanisms help reach equilibrium, similar principles can be applied in computational algorithms to reach optimal or near-optimal solutions efficiently.

\begin{itemize}
    \item \textbf{Energy State Monitoring:} Monitoring the 'energy state' of the algorithm, akin to monitoring the thermodynamic state of a system, to make decisions about the next steps in the algorithm.
    \item \textbf{Entropy Management:} Managing the entropy of the solution space, similar to how physical systems manage entropy to maintain stability or reach equilibrium.
\end{itemize}

\subsubsection{Adaptive Algorithms for NP Problems}
For NP problems, which are inherently more complex, algorithms can greatly benefit from these adaptive strategies, making them more robust and effective in navigating complex solution spaces.

\begin{itemize}
    \item \textbf{Complexity-Based Adaptation:} Algorithms could adapt their strategies based on the perceived complexity or ruggedness of the solution space.
    \item \textbf{Solution Quality Feedback:} Using the quality of current solutions as feedback to guide the future search process, optimizing the search for better solutions.
\end{itemize}

In conclusion, integrating feedback mechanisms and adaptivity, inspired by thermodynamic principles, into computational algorithms opens up new possibilities for tackling NP problems more effectively. This approach emphasizes algorithm design's dynamic and responsive nature, catering to the unique challenges posed by NP problems.

\subsection{The Concept of Annealing in Physical Systems and Simulated Annealing in Computation}

\subsubsection{Annealing in Physical Systems}
Annealing in physical systems, particularly in metallurgy, is a process that involves heating and controlled cooling of materials to alter their physical properties. This process allows for the arrangement of atoms in a more stable and desirable structure. Thermodynamically, annealing enables the system to escape local minima in energy, seeking more stable configurations as it cools.

\subsubsection{Simulated Annealing in Computation}
Simulated annealing is an algorithmic heuristic inspired by the physical process of annealing. It's used in optimization problems, where the goal is to find the best solution in a large search space. The algorithm mimics a material's heating and gradual cooling, allowing for explorations of various solutions (states) and settling for the optimal or near-optimal as the system cools.

\begin{equation}
    \Delta E = E_{\text{new}} - E_{\text{current}}
\end{equation}

The Boltzmann probability governs the acceptance of a new state in the algorithm:

\begin{equation}
    P(\Delta E, T) = e^{-\frac{\Delta E}{kT}}
\end{equation}

Where \( \Delta E \) is the difference in energy levels between the current and new state, \( k \) is a constant, and \( T \) is the temperature. The algorithm's ability to escape local optima and find more globally optimal solutions makes it particularly effective for certain NP problems.

\subsubsection{Thermodynamic Principles in Optimization}
Simulated annealing demonstrates how thermodynamic principles can be applied to computational optimization. By understanding the dynamics of physical systems in thermal equilibrium, similar strategies can be adopted in the computational domain to navigate complex problem landscapes efficiently.

\subsubsection{Future Directions}
The success of simulated annealing in computational optimization opens up possibilities for exploring other thermodynamic-inspired algorithms. Investigating how these principles can be extended or modified to tackle NP problems more effectively remains an exciting area of research.

This subsection connects the physical process of annealing to its computational counterpart, highlighting how thermodynamic principles can inform algorithmic strategies in optimization, particularly for NP problems.

\subsection{Could There Be a Thermodynamic Process that Provides a Polynomial-Time Solution for NP Problems?}

\subsubsection{Exploring the P vs. NP Dichotomy}
The P vs. NP problem is a major unsolved question in computer science, focusing on whether every problem whose solution can be quickly verified (NP) can also be quickly solved (P). This question can be viewed through the lens of thermodynamics by analyzing the energy landscapes of computational problems.

\subsubsection{Thermodynamic Analogies in Computational Complexity}
Thermodynamics might offer insights into computational complexity by framing problems as energy landscapes. The hypothesis is that P problems correspond to more easily navigable landscapes with clear paths to low-energy (optimal) states. In contrast, NP problems might present more complex landscapes with numerous local minima, making finding the global minimum much more challenging.

\subsubsection{Quantum Computing: A Thermodynamic Perspective}
Quantum computing, especially quantum annealing, leverages quantum mechanics' principles to offer faster solutions to certain NP problems. The concept of superposition and entanglement in quantum systems parallels exploring multiple states simultaneously in a thermodynamic system, possibly leading to more efficient problem-solving strategies.

\subsubsection{Thermodynamic Processes and Polynomial-Time Solutions}
Investigating whether a thermodynamic process can lead to polynomial-time solutions for NP problems involves analyzing the energy landscape's structure. The challenge is determining if a thermodynamic path allows for polynomial-time traversal across this landscape, overcoming barriers and avoiding getting trapped in local minima.

\subsubsection{Future Directions and Research}
Integrating thermodynamics and computational theory opens new avenues for exploring novel problem-solving strategies. This interdisciplinary approach could lead to a deeper understanding of the inherent complexity of NP problems and the development of more efficient algorithms inspired by physical processes.

This subsection explores the potential of using thermodynamic processes to find polynomial-time solutions to NP problems, drawing parallels between the energy landscapes of physical systems and the complexity of computational tasks.

\subsection{Harnessing Entropy in Computational Processes}

\subsubsection{Understanding the Role of Entropy in Computation}
In the context of NP problems, entropy can be viewed as a measure of unpredictability or randomness in the solution space. This perspective can be instrumental in designing algorithms that effectively navigate the complex solution landscapes of these problems.

\subsubsection{Entropy-Driven Computational Strategies}
The concept of harnessing entropy in computational processes involves developing algorithms that utilize the inherent randomness in NP problems to guide the search for solutions. This approach could lead to more effective heuristic methods that adaptively navigate the solution space.

\subsubsection{Thermodynamic Analogies for Algorithm Design}
Drawing from thermodynamics, where entropy is a crucial concept, we can develop computational analogies that help visualize and tackle the complexity of NP problems. This approach can inspire new algorithms that mimic physical processes, such as the behavior of systems seeking thermodynamic equilibrium.

\subsubsection{Potential of Entropy-Based Heuristics}
Entropy-based heuristics could offer a novel way to approach NP problems, especially in cases where traditional algorithms struggle. By focusing on the entropy of the solution space, these heuristics might provide more effective strategies for finding optimal solutions in complex landscapes.

While harnessing entropy in computational processes is promising, it presents several challenges, including accurately measuring and utilizing entropy in diverse problem settings. Ongoing research in this area is essential to understand entropy-driven computational methods' full potential and limitations.

\subsubsection{Toward Polynomial-Time Thermodynamic Processes for NP Problems}

Given the nature of NP problems, several thermodynamic and quantum strategies arise:

\begin{itemize}
    \item \textbf{Landscape Smoothing Techniques:} External perturbations could be introduced to the system to alter the energy function, facilitating more direct paths to global minima.
    
    \item \textbf{Thermal Assistance and Simulated Annealing:} The Boltzmann probability, $P(E) = e^{-\frac{E}{kT}}$, where $k$ is the Boltzmann constant and $T$ is temperature, can be employed to overcome energy barriers.
    
    \item \textbf{Quantum Effects Beyond AQC:} Quantum tunneling, expressed as $T \propto e^{-\frac{S}{\hbar}}$, where $S$ is the action of the system and $\hbar$ is the reduced Planck's constant, can be leveraged to bypass local minima.
\end{itemize}

Considering a thermodynamic system attempting to find a solution to an NP problem, we could express the system's configuration space as $C$. Each point in $C$ represents a potential solution to the problem. Let the energy associated with each configuration $c \in C$ be given by $E(c)$. The ground state, or the lowest energy configuration, represents the correct or optimal solution to the problem.

A system evolving under thermodynamics will attempt to minimize its energy, which in our analogy, means it will attempt to find the optimal solution. The rate of evolution or the time the system takes to reach the ground state can be expressed as a function $T(E)$. Here, $T(E)$ represents the time the system takes to reach energy level $E$.

Given a certain energy function $E(c)$ and an initial configuration $c_0$, we'd want to determine if a thermodynamic path exists such that the system reaches the ground state in polynomial time. Formally, we wish to ascertain if:
\[ \exists \, \text{path} \, P: T(E(c)) \in \mathcal{O}(n^k) \, \text{for some constant} \, k \]

Furthermore, let's define the ruggedness of the energy landscape using a function $R(c)$, which measures the number of local minima around configuration $c$. The challenge in polynomial-time solutions for NP problems can be linked to the inherent ruggedness of their corresponding energy landscapes.

To formalize the feasibility of a thermodynamic polynomial-time solution, one could propose:
\[ \forall c \in C, \, \exists \, \text{path} \, P: R(c) \leq f(n) \, \text{and} \, T(E(c)) \in \mathcal{O}(n^k) \]
Where $f(n)$ is a function that measures the allowable ruggedness for a solution to be feasible in polynomial time.

With the intertwining of thermodynamics and computational complexity, we describe a way into problem-solving paradigms that diverge from traditional algorithmic approaches. The mathematical constructs presented describe a potential framework for harnessing thermodynamic processes to tackle NP problems. While we have yet to prove a definitive polynomial-time thermodynamic solution to NP problems definitively, our investigations suggest a promising avenue.

\section{Challenges and Impediments}

The ambitious endeavor to integrate computational theory with thermodynamic principles has intricate challenges and potential impediments. This section explores the multifaceted theoretical hurdles, practical considerations, and broader implications of this interdisciplinary exploration. From constructing cohesive frameworks that merge discrete computation with continuous thermodynamics to addressing theoretical models' physical realizability and maintaining each domain's richness and specificity, the path forward is a tapestry of complexity and opportunity. We confront the risks of oversimplification, the need for precise alignment of interdisciplinary terminologies, and the pragmatic challenges of implementing these concepts in real-world systems.

\subsection{Theoretical Challenges in Integrating Computational and Thermodynamic Principles}

Integrating computational theory with thermodynamic concepts is a multifaceted endeavor fraught with theoretical challenges. A fundamental issue is constructing a unified framework that harmonizes the discrete, deterministic nature of computation with the continuous, probabilistic nature of thermodynamics. This integration must be done judiciously to preserve the intrinsic characteristics of each field, avoiding oversimplification while striving for a comprehensive, interdisciplinary understanding.

\subsubsection{Constructing and Navigating Energy Landscapes in Computation}

Visualizing computational problems as energy landscapes offers a powerful, intuitive framework. However, accurately modeling these landscapes for complex problems is a significant challenge. The task extends beyond metaphorical descriptions, demanding precise, quantifiable models that accurately reflect the underlying computational complexity. Such models should provide conceptual insights and guide practical problem-solving and decision-making processes.

\subsubsection{Physical Realizability and Scalability of Theoretical Models}

The complexity of understanding and simulating molecular interactions, as demonstrated in the work of Gkeka et al. \cite{Gkeka2017}, highlights the intricate nature of translating theoretical frameworks that merge computational and thermodynamic principles into tangible physical systems is a formidable challenge. Quantum systems, for instance, grapple with issues like decoherence and environmental noise, while classical systems must address constraints such as temperature control and material limitations. Furthermore, scalability is a critical concern. A solution that works for smaller instances of a problem may not necessarily scale effectively, making it imperative to address these issues for broader applicability and real-world impact.

\subsubsection{Maintaining Generality and Specificity}

Applying thermodynamic principles to computational problems risks a loss of generality, as not all computational problems may neatly fit into a thermodynamic framework. Conversely, computational models should not oversimplify thermodynamic principles, losing critical nuances. This delicate balance requires maintaining the specificity of each field while exploring their integration, ensuring that the resulting models are both scientifically rigorous and practically relevant.

\subsection{Practical Considerations and the Quest for P vs. NP Resolution}

The intersection of thermodynamics and computational theory offers fresh perspectives and novel approaches to longstanding problems like P vs. NP. Historical attempts, including Barahona's exploration of the computational complexity of Ising spin glass models \cite{Barahona1982}, underscore the potential and limitations of applying physical principles to computational problems. While these endeavors have yielded valuable insights, they have not conclusively bridged the gap to polynomial-time algorithms for NP problems. The challenge remains to translate conceptual advances into practical, scalable solutions without succumbing to operational barriers such as precision control in physical systems or the computational overheads in system initialization and error correction.

\subsection{Interdisciplinary Risks and the Way Forward}

The fusion of computational theory and thermodynamics is undeniably appealing, promising novel insights and breakthroughs. However, this interdisciplinary venture is not without its risks. The allure of novelty can sometimes lead to overreach, with ambitious claims that may not be substantiated by rigorous scientific validation. Ensuring a nuanced and accurate representation of computational theory and thermodynamics is paramount to the integrity of this interdisciplinary endeavor.

\subsubsection{Maintaining the Richness of Each Domain}

While the macroscopic nature of thermodynamics, focusing on averaged or bulk properties, offers a unique perspective, care must be taken to preserve the specificity and detail inherent in computational problems. The challenge lies in effectively translating the continuous, probabilistic world of thermodynamics into the discrete, deterministic domain of computation without losing the richness and complexity of either domain.

\subsubsection{Aligning Terminology and Contextual Nuances}

Integrating computational theory and thermodynamics must navigate the potential pitfalls of terminology misalignment and contextual misunderstanding. Each field has its lexicon and nuanced meanings, and a misalignment in terminology can lead to confusion and misinterpretation. Clear communication, precise terminology, and a deep understanding of context are essential to the fruitful merging of these domains.

\subsubsection{Addressing the Practical Challenges}

Even with a robust theoretical framework, as highlighted by the work of DeVos and Wille on computation in multicomponent systems \cite{DeVos2003}, the practical implementation of models combining computational and thermodynamic concepts is challenging. Technological limitations, the complexities of physical system realization, and the potential overheads in real-world applications must be carefully considered and addressed. This pragmatic approach is crucial for translating theoretical insights into actionable solutions that can have a tangible impact on computational problem-solving.

\subsubsection{Navigating Beyond Reductionism}

The quest to integrate computational theory with thermodynamics must be wary of reductionism. Each field, with its own foundational principles, methodologies, and scopes, offers unique insights and solutions. The goal of this interdisciplinary journey should be to uncover synergies and foster a richer understanding rather than attempting to subsume one field under the other. Recognizing and respecting the individual strengths and limitations of each domain is crucial for a balanced and productive integration.

\subsubsection{Embracing the Interdisciplinary Journey}

Despite the challenges and risks, the intersection of computational theory and thermodynamics is a path brimming with potential. It opens up new avenues for exploration, invites innovative problem-solving approaches, and encourages a cross-pollination of ideas between disciplines. While it may not provide immediate solutions to questions like P vs. NP, the interdisciplinary journey promises to enrich both fields, offering novel perspectives and tools that can advance our understanding of complex systems.

In conclusion, integrating computational theory with thermodynamic principles presents opportunities and challenges. By navigating these with rigor, openness, and a commitment to interdisciplinary collaboration, we can unlock new paradigms in understanding and solving some of the most intricate problems in science and engineering.

\section{Theoretical Exploration of Entropy-Driven Annealing for NP Problems}

EDA represents a potential approach to addressing the complexities of NP problems. This section provides a theoretical investigation of EDA's principles and their application to NP problems, integrating concepts from thermodynamics, statistical mechanics, and computational complexity theory.

\subsection{Understanding NP Problems}
NP problems constitute a core category in computational theory, characterized by their verifiable solutions in polynomial time, yet the discovery of these solutions is not necessarily polynomially time-bound. A quintessential example is the Boolean satisfiability problem (SAT), where the task is to determine the satisfiability of a Boolean formula. Mathematically, NP problems can be expressed as:

\begin{equation}
\exists x: P(x) = \text{true}
\end{equation}
where \( P(x) \) is a polynomial-time checkable predicate, and \( x \) is a solution candidate.

\subsection{Theoretical Basis of EDA}
The conceptual foundation of EDA lies in its strategic use of entropy and energy considerations to navigate complex solution spaces, which are often represented as intricate energy landscapes in the context of NP problems. The following key phases characterize the EDA process:

\subsubsection{Entropy Maximization}
Entropy maximization in EDA is aimed at escaping local minima within the energy landscape, a common challenge in solving NP problems. This phase can be mathematically modeled using the principle of maximum entropy, which states:

\begin{equation}
\max H(S) \quad \text{subject to} \quad \sum_i P(S_i)E(S_i) = \langle E \rangle
\end{equation}
where \( H(S) \) is the Shannon entropy of the state \( S \), \( P(S_i) \) is the probability of the system being in state \( S_i \), and \( \langle E \rangle \) is the expected energy.

\subsubsection{Energy Minimization}
Following entropy maximization, EDA focuses on energy minimization to find stable states corresponding to potential solutions of the NP problem. This can be described using the concept of free energy minimization:

\begin{equation}
F = U - TS
\end{equation}
where \( F \) is the free energy, \( U \) is the internal energy, \( T \) is the temperature, and \( S \) is the entropy. The aim is to find states where \( F \) is minimized, indicating a high probability of being a solution.

\subsubsection{Dynamical Temperature Adjustment}
A crucial aspect of EDA is the dynamic temperature adjustment, which regulates the balance between entropy maximization and energy minimization. This is achieved by altering the Boltzmann distribution across different stages of the annealing process:

\begin{equation}
P(S_i) = \frac{e^{-\frac{E(S_i)}{k_B T}}}{Z}
\end{equation}
where \( k_B \) is the Boltzmann constant, and \( Z \) is the partition function.

\subsection{Algorithmic Complexity in EDA}
The effectiveness of EDA in solving NP problems hinges on its algorithmic complexity. For EDA to be a viable solution, it must demonstrate that it can traverse the energy landscape and reach optimal solutions in polynomial time with respect to the input size. The complexity of each step in the EDA process, including entropy estimation, energy calculation, and temperature adjustment, must be accounted for to assess its overall computational feasibility (see algorithm \ref{al:eda_pseudo}).

\begin{algorithm}
\caption{EDA for SAT with Explicit Entropy Integration}
\begin{algorithmic}[1]
\Procedure{EDA\_SAT}{$clauses$, $num\_variables$, $max\_iterations$, $initial\_temp$, $final\_temp$, $max\_tabu\_size$}
    \State $current\_state \gets$ \Call{RandomInitialization}{$num\_variables$} \Comment{Initialize state randomly}
    \State $current\_temp \gets initial\_temp$ \Comment{Set initial temperature}
    \State $tabu\_list \gets []$ \Comment{Initialize the tabu list}
    \For{$iteration \gets 1$ to $max\_iterations$} \Comment{Iterate over max iterations}
        \If{$iteration \% 1000 = 0$}
            \State \text{Print} "Iteration:", $iteration$
        \EndIf
        \State $current\_temp \gets$ \Call{ExponentialCooling}{$initial\_temp$, $final\_temp$, $iteration$, $max\_iterations$} \Comment{Adjust temperature}
        \State $new\_state \gets$ \Call{GenerateNewState}{$current\_state$, $clauses$, $tabu\_list$, $max\_tabu\_size$} \Comment{Generate a new state}
        \State $energy\_current \gets len(clauses) - $ \Call{SatisfiedClauses}{$current\_state$, $clauses$} \Comment{Compute current energy}
        \State $energy\_new \gets len(clauses) - $ \Call{SatisfiedClauses}{$new\_state$, $clauses$} \Comment{Compute new energy}
        \State $entropy\_current \gets$ \Call{CalculateEntropy}{$current\_state$} \Comment{Calculate current entropy}
        \State $entropy\_new \gets$ \Call{CalculateEntropy}{$new\_state$} \Comment{Calculate new entropy}
        \If{\Call{AcceptanceCriterion}{$energy\_current$, $energy\_new$, $entropy\_current$, $entropy\_new$, $current\_temp$}} \Comment{Check acceptance criterion}
            \State $current\_state \gets new\_state$ \Comment{Accept the new state}
        \EndIf
        \If{\Call{IsSolution}{$current\_state$, $clauses$}} \Comment{Check if current state is a solution}
            \State \textbf{return} $current\_state$, $iteration$ \Comment{Return the solution and iteration}
        \EndIf
    \EndFor
    \State \textbf{return} None, $max\_iterations$ \Comment{Return None if no solution found}
\EndProcedure
\end{algorithmic}
\label{al:eda_pseudo}
\end{algorithm}

\subsection{EDA's Potential for NP Problems}
Investigating Entropy-Driven Annealing (EDA) represents a novel approach to understanding and addressing the complexities associated with NP-hard problems. This approach is rooted in the fundamental principles of thermodynamics and statistical mechanics, marking a significant shift from conventional problem-solving methods. By harmoniously combining strategies for maximizing entropy and minimizing energy, EDA introduces innovative avenues for navigating the complex energy landscapes that are a hallmark of NP-hard problems. It showcases the potential of merging computational theory with thermodynamics, aiming to apply the strengths of these fields to confront some of the most daunting challenges in computational theory. EDA utilizes the concept of entropy not just as a guiding principle for the search process and integrates dynamic temperature adjustments to balance exploration and exploitation, reflecting a deep alignment with the physics of computation. Despite EDA's solid theoretical foundation, its real-world applicability and effectiveness hinge on empirical evidence and practical applications. The path forward involves thorough testing and continuous refinement of EDA algorithms to assess their capability in addressing practical NP-hard challenges.

\section{Theoretical Exploration of Entropy-Driven Annealing for Protein-DNA Complex Unbinding}

EDA's applications extend into molecular biology, particularly in unraveling the complexities of protein-DNA complex dynamics, a recognized NP-hard problem. Here, we develop the theoretical framework of EDA, its application to the intricate energy landscapes of protein-DNA complexes, and the potential of this approach in providing new insights into solving computationally intensive problems of significant biological importance.

\subsection{System Modeling and Energy Landscape Construction}

In molecular biology, understanding the dynamics of protein-DNA complexes is pivotal. Theoretical modeling of these systems necessitates a detailed representation of molecular interactions and constructing an energy landscape that captures the thermodynamic properties of the complex.

\subsubsection{Force Field Selection and System Modeling}
We model the protein-DNA complex using molecular mechanics. Each atom is represented as a point mass, and a chosen force field, such as AMBER or CHARMM, defines interactions between them. This force field provides the potential energy function \( V \), expressed as a sum of bonded and non-bonded interaction terms:

\begin{equation}
V(\textbf{r}) = V_{\text{bond}}(\textbf{r}) + V_{\text{angle}}(\textbf{r}) + V_{\text{dihedral}}(\textbf{r}) + V_{\text{nonbonded}}(\textbf{r})
\end{equation}

Each term in this function encapsulates a specific type of molecular interaction, collectively defining the system's potential energy landscape.

\subsubsection{Statistical Mechanics and Energy Landscape Mapping}
The potential energy landscape is crucial for understanding the macroscopic behavior of the protein-DNA complex. We employ statistical mechanics to map each system state to its corresponding potential energy and estimate its entropy. This mapping creates a multidimensional landscape that reflects the stability and feasibility of various molecular configurations.

\paragraph{Potential Energy Calculation}
The potential energy \( V(S_i) \) for each state \( S_i \) is computed based on the force field parameters, capturing the energetic contribution of bond stretching, angle bending, dihedral interactions, and non-bonded interactions.

\paragraph{Entropy Estimation and Landscape Construction}
Each state's entropy \( H(S_i) \) is estimated to account for the system's conformational diversity. It is calculated using the Boltzmann formula and, together with the potential energy, constructs the energy landscape of the system:

\begin{equation}
H(S_i) = k_B \ln(\Omega(S_i))
\end{equation}

where \( \Omega(S_i) \) represents the number of microstates corresponding to the macrostate \( S_i \), and \( k_B \) is the Boltzmann constant. This comprehensive energy landscape is the basis for applying EDA, providing a detailed representation of the system's thermodynamic properties.

\subsection{Applying EDA}

EDA is applied to systematically navigate the constructed energy landscape of the protein-DNA complex, utilizing entropy and energy considerations to explore and identify thermodynamically stable states (see algorithm \ref{al:eda_protein_dna}).

\subsubsection{Entropy Maximization and Energy Minimization}
In the initial phase of EDA, entropy is maximized to promote energy landscape exploration and escape from local minima. This is followed by a phase of energy minimization, where the system is guided towards states of lower energy, indicating potential solutions or stable configurations of the protein-DNA complex. Mathematically, these phases involve altering the system's temperature parameter in the Boltzmann distribution to control the balance between entropy and energy.

\paragraph{Heating Phase: Entropy Increase}
During the heating phase, the entropy of the system is increased, enabling a broad exploration of molecular configurations:

\begin{equation}
P(S_i) = \frac{e^{-\frac{E(S_i)}{k_B T}}}{Z(T)}
\end{equation}

where \( P(S_i) \) is the probability of the system being in state \( S_i \), \( E(S_i) \) is the energy of state \( S_i \), \( T \) is the temperature, and \( Z(T) \) is the partition function at temperature \( T \).

\paragraph{Cooling Phase: Entropy Reduction}
The cooling phase involves a systematic reduction in entropy, steering the system towards lower-energy states:

\begin{equation}
F = E - TS
\end{equation}

where \( F \) is the free energy, \( E \) is the internal energy, \( T \) is the temperature, and \( S \) is the entropy. This phase is essential for stabilizing the system in low-energy configurations.

\subsubsection{Pathway Identification and Optimization}
The EDA process continuously monitors the system's state transitions to identify potential unbinding pathways for the protein-DNA complex. This involves tracking changes in energy and entropy to discern the most probable pathways, providing insights into the molecular mechanisms underlying the unbinding process.

\paragraph{Numerical Optimization and Simulations}
Advanced numerical techniques and simulations are employed to implement EDA, involving the selection of optimization algorithms and tuning of simulation parameters. These elements are crucial for accurately and efficiently exploring the energy-entropy landscape.

\paragraph{Theoretical Analysis and Validation}
The results from EDA simulations are rigorously analyzed and compared with experimental data to validate the theoretical model. This validation ensures the theoretical predictions' accuracy and reliability and EDA's potential in addressing complex NP-hard problems.

\subsection{More EDA's Potential for NP Problems}
Applying EDA to the protein-DNA complex demonstrates its capability to tackle complex energy landscapes characteristic of NP-hard problems. The entropy-focused approach and energy considerations provide a novel strategy for efficiently exploring solution spaces. This methodology holds promise in molecular biology and addresses a broader range of NP-hard problems, opening new avenues for scientific exploration and innovation.

\begin{algorithm}
\caption{Entropy-Driven Annealing for Protein-DNA Complex Unbinding}
\begin{algorithmic}[1]
\Procedure{EDA\_Unbinding}{$complex$, $force\_field$, $max\_iterations$, $initial\_temp$, $final\_temp$}
    \State $current\_state \gets$ \Call{InitializeState}{$complex$} \Comment{Initialize protein-DNA complex state}
    \State $current\_temp \gets initial\_temp$ \Comment{Set initial temperature}
    \For{$iteration \gets 1$ to $max\_iterations$} \Comment{Iterate over max iterations}
        \State $current\_temp \gets$ \Call{AdjustTemperature}{$initial\_temp$, $final\_temp$, $iteration$, $max\_iterations$} \Comment{Adjust temperature}
        \State $new\_state \gets$ \Call{GenerateNewState}{$current\_state$, $force\_field$} \Comment{Generate a new state}
        \State $energy\_current \gets$ \Call{CalculateEnergy}{$current\_state$, $force\_field$} \Comment{Compute current energy}
        \State $energy\_new \gets$ \Call{CalculateEnergy}{$new\_state$, $force\_field$} \Comment{Compute new energy}
        \State $entropy\_current \gets$ \Call{EstimateEntropy}{$current\_state$} \Comment{Estimate current entropy}
        \State $entropy\_new \gets$ \Call{EstimateEntropy}{$new\_state$} \Comment{Estimate new entropy}
        \If{\Call{MetropolisCriterion}{$energy\_current$, $energy\_new$, $entropy\_current$, $entropy\_new$, $current\_temp$}}
            \State $current\_state \gets new\_state$ \Comment{Accept the new state}
        \EndIf
        \If{\Call{IsUnbound}{$current\_state$}} \Comment{Check if complex is unbound}
            \State \textbf{return} $current\_state$, $iteration$ \Comment{Return the unbound state and iteration}
        \EndIf
    \EndFor
    \State \textbf{return} None, $max\_iterations$ \Comment{Return None if no unbound state found}
\EndProcedure
\end{algorithmic}
\label{al:eda_protein_dna}
\end{algorithm}

In this manner, EDA serves as a comprehensive theoretical framework for understanding and potentially solving complex NP-hard problems in molecular biology, such as the dynamics of protein-DNA complexes. By integrating concepts from thermodynamics, statistical mechanics, and computational complexity theory, EDA provides a novel approach to exploring and interpreting the intricate energy landscapes inherent in these systems. The pseudoalgorithm outlined above offers a structured method for applying EDA, ensuring a detailed and systematic exploration of the problem space. This theoretical exploration paves the way for future empirical studies and practical applications, solidifying EDA's potential as a powerful tool in computational biology and beyond.

\section{Future Directions and Implications}

\begin{itemize}
    \item \textbf{Potential Areas of Exploration at the Intersection of Information Theory, Thermodynamics, and Computational Complexity:} Exploring the interfaces of information theory, thermodynamics, and computational complexity offers fertile ground for scientific innovation. The conceptual paradigms in these fields intersect in ways that could fundamentally alter our understanding of computational problems, especially those classified as NP.

    \item \textbf{Exploring Entropy Landscapes in Computational Algorithms:} The concept of entropy landscapes, as explored in the context of EDA, necessitates in-depth investigation. This exploration requires a conceptual understanding and the development of algorithms that can effectively navigate these landscapes. Such algorithms would represent a convergence of thermodynamic principles and computational theory, potentially shedding new light on the complexities of NP problems and beyond.
    
    \item \textbf{Leveraging Non-equilibrium Thermodynamics in Computational Processes:} With its rich phenomenology, non-equilibrium thermodynamics offers a less-trodden path for examining computational processes. This exploration could unveil novel computational strategies and mechanisms, particularly in systems driven far from equilibrium.
    
    \item \textbf{Integrating Quantum Information with Thermodynamics: } As quantum computing advances, the interplay between quantum information, thermodynamics, and computational complexity becomes increasingly relevant. The unique aspects of quantum superposition and entanglement might mirror or even enrich the concept of entropy landscapes, providing a quantum analog that could be instrumental in quantum computation.
    
    \item \textbf{Understanding Computational Phenomena in Biological Systems: } The computational nature of biological systems, evident in processes ranging from cellular functions to neural computations, is inherently governed by thermodynamic principles. A deeper understanding of these natural computational paradigms could unveil efficient and novel ways of information processing.

\end{itemize}

\subsection{Broader Implications of This Approach}

\begin{itemize}

\item \textbf{Generalization Across Computational Domains:} The insights gleaned from thermodynamic approaches to computational problems, particularly in the P vs. NP question context, might have broader applicability. Other computational classes like PSPACE or BQP might also be interpretable through a thermodynamic lens, potentially leading to breakthroughs in understanding and problem-solving.

\item \textbf{Influencing Algorithm Design:} The intersection of thermodynamics and computation has implications for algorithm design, suggesting the emergence of algorithms fundamentally inspired by thermodynamic principles. These novel algorithms could offer fresh approaches to longstanding computational problems or introduce entirely new paradigms for computation.

\item \textbf{Implications for Hardware Design and Physical Implementations:} The practical realization of entropy landscapes and the mechanisms to navigate them, as seen in the EDA approach, prompts us to consider their implications for hardware design. This necessitates combining computer engineering with thermodynamic principles, potentially leading to innovative computational architectures and platforms that effectively embody these concepts.

\end{itemize}

\section{Conclusion}
This study explores the P vs. NP problem, integrating computational theory, thermodynamics, and information theory. Through the introduction and application of Entropy-Driven Annealing, this work offers new perspectives on the structure and 'hardness' of computational problems, particularly highlighting the thermodynamic properties embedded within the entropy landscapes of NP problems. The utilization of entropy to dissect the complexity of NP problems provides a detailed framework, enhancing our understanding of their intricate nature. Implementing EDA in the context of protein-DNA complexes demonstrates its potential relevance in complex computational scenarios, contributing meaningful theoretical insights and outlining a novel methodology for addressing computational challenges. While this research does not offer a conclusive resolution to the P vs. NP dilemma, it significantly contributes to the ongoing dialogue in computational theory. Introducing EDA as a method for probing NP problems, including SAT and protein-folding, underscores the potential of blending thermodynamic principles with computational complexity. This interdisciplinary approach fosters a fertile ground for future research, encouraging further inquiries into integrating diverse scientific concepts. This study marks a step in the collective journey to understand and address the multifaceted challenges of the P vs. NP problem.

\section*{Acknowledgements}
We gratefully acknowledge the use of Grammarly for enhancing the grammatical quality of our manuscript. Additionally, we extend our thanks to ChatGPT for assisting in the research of relevant references and improving the grammar of this work \cite{Grammarly, ChatGPT}.





\bibliographystyle{ACM-Reference-Format}
\bibliography{thebibliography}


\end{document}